\documentclass{article}

\usepackage{PRIMEarxiv}

\usepackage[utf8]{inputenc}
\usepackage[T1]{fontenc}
\usepackage[table,xcdraw]{xcolor}
\usepackage{graphicx}
\usepackage{booktabs}
\usepackage{tabularx}
\usepackage{array}
\usepackage{amsmath}
\usepackage{amsfonts}
\usepackage{amssymb}
\usepackage{amsthm}
\usepackage{mathtools}
\usepackage{nicefrac}
\usepackage{microtype}
\usepackage{url}
\usepackage{cite}
\usepackage{tikz}
\usepackage{pgfplots}
\usepackage{subcaption}
\usepackage{enumitem}
\usepackage{hyperref}

\pgfplotsset{compat=1.18}
\usetikzlibrary{calc,positioning,fit,backgrounds,arrows.meta}
\graphicspath{{./}{figures/}}
\numberwithin{equation}{section}
\setlist[enumerate,1]{label=(\arabic*),ref=(\arabic*)}

\newtheorem{theorem}{Theorem}[section]
\newtheorem{lemma}[theorem]{Lemma}
\newtheorem{proposition}[theorem]{Proposition}
\newtheorem{corollary}[theorem]{Corollary}

\theoremstyle{definition}
\newtheorem{definition}[theorem]{Definition}
\newtheorem{assumption}[theorem]{Assumption}

\newtheorem{example}[theorem]{Example}

\theoremstyle{remark}
\newtheorem{remark}[theorem]{Remark}

\title{Exit-and-Join Dynamics and Equilibrium in Continuum Cooperative Games}
\author{
  Quanyan Zhu \\
  Department of Electrical and Computer Engineering \\
  New York University Tandon School of Engineering \\
  Brooklyn, NY, USA \\
  \texttt{quanyan.zhu@nyu.edu}
}

\begin{document}

\maketitle

\begin{abstract}
This paper develops a continuum theory of exit-and-join coalition
dynamics in nonatomic cooperative games. We first extend the
Aumann-Shapley value and the Aumann-Dr\`eze value to coalition
structures in which each coalition is a nonatomic restricted game.
The resulting payoff density gives a marginal-contribution-based
incentive for agents to remain in, exit, or join coalitions. We then
derive deterministic mean-field dynamics from decentralized switching
rules and show that payoff-difference switching includes replicator
dynamics as a special case. The associated exit-and-join equilibrium
is characterized by the absence of profitable positive-mass deviations
and is equivalent to stationarity of the induced mass dynamics under
incentive-compatible and strictly payoff-responsive switching rates. For
mass-based cooperative games, we construct a Lyapunov function and obtain
global convergence under strict concavity.
We further connect the equilibrium concept to Wardrop equilibria and
variational inequalities, and we extend the model to switching costs
and endogenous acceptance constraints. The framework links cooperative
value allocation, noncooperative mobility, and evolutionary selection
in large-population coalition systems.
\end{abstract}


\section{Introduction}
\label{sec:introduction}

Cooperative game theory studies how groups of agents create value and how
that value should be allocated among participants. The classical Shapley
value \cite{Shapley1953,winter2002shapley} provides a canonical allocation
rule for finite-player games by averaging marginal contributions over all
orders of entry. In large populations, however, individual agents are
negligible. A single agent has zero measure and cannot create a discrete
jump in coalition value. The Aumann-Shapley value replaces finite
permutations by infinitesimal marginal contributions along continuous
participation paths \cite{aumann1975values,aumann2015values}.

Coalition structures add a second layer to this problem. In many economic,
organizational, and multi-agent systems, agents are not arranged in one
grand coalition but are distributed across several coalitions. The
Aumann-Dr\`eze value \cite{aumann1974cooperative,AumannDreze1974}
addresses this issue in finite games by applying the Shapley principle
within each coalition of a fixed coalition structure. The first objective
of this paper is to develop the corresponding construction for nonatomic
cooperative games. Each coalition is treated as a restricted nonatomic
game, and the payoff assigned to an agent is its Aumann-Shapley marginal
contribution within that coalition.

The second objective is dynamic. Coalition structures are rarely fixed in
applications: workers move across firms, agents reallocate across tasks,
users migrate across platforms, and alliances form or dissolve in response
to changing incentives. We therefore introduce exit-and-join dynamics in
which agents compare the payoff density in their current coalition with
the payoff densities available elsewhere. A positive-mass flow from one
coalition to another occurs only when the destination offers a strict
payoff improvement, possibly subject to switching costs or acceptance
constraints. Since the player space is nonatomic, the primitive individual
switches aggregate into deterministic mean-field dynamics for coalition
masses.

The resulting model connects three viewpoints. First, it is cooperative:
the payoff field is generated by a transferable-utility cooperative game
and by continuum Aumann-Dr\`eze values. Second, it is noncooperative:
agents move unilaterally whenever another coalition offers a better payoff.
Third, it is evolutionary: mass grows in coalitions with higher payoff and
shrinks in coalitions with lower payoff, yielding replicator-type dynamics
under standard payoff-difference switching rules.

\noindent\textbf{Main contributions.}
The paper makes the following contributions.
\begin{enumerate}
\item We formulate nonatomic cooperative games, participation paths,
marginal contribution densities, and the Aumann-Shapley value in a form
suited for coalition structures.
\item We define the continuum Aumann-Dr\`eze value by restricting the game
to each coalition and applying the Aumann-Shapley construction inside the
restricted nonatomic measure space.
\item We derive exit-and-join mass dynamics from decentralized switching
intensities and show that payoff-difference switching recovers the
replicator equation.
\item We characterize exit-and-join equilibrium and prove its equivalence
with stationarity of the dynamics under incentive-compatible and strictly
payoff-responsive switching rates.
\item For mass-based games, we construct a Lyapunov function and establish
global convergence under strict concavity.
\item We show that the equilibrium condition coincides with Wardrop
equilibrium in the induced nonatomic population game, and we give the
corresponding variational inequality formulation.
\item We extend the framework to switching costs and endogenous acceptance
rules, which produce state-dependent feasible deviation cones and a
quasi-variational inequality structure.
\item We interpret exit-and-join dynamics as an evolutionary selection
process over cooperative coalition structures.
\end{enumerate}

\noindent\textbf{Related work.}
The paper builds on the classical Shapley value \cite{Shapley1953} and
its nonatomic extension by Aumann and Shapley
\cite{aumann1975values,aumann2015values}. Coalition structures and
coalitional values trace back to Aumann and Dr\`eze
\cite{aumann1974cooperative,AumannDreze1974}. Dynamic cooperative games
and coalition formation have been studied in several directions, including
dynamic cooperative games \cite{filar2000dynamic,bauso2009robust},
coalition formation processes \cite{apt2009generic,konishi2003coalition},
and computational or learning-based approaches
\cite{hart1996bargaining,faigle2001computation,touati2021bayesian}.
The population-dynamic part of the paper is related to evolutionary game
theory and population games \cite{vincent2000evolution,
vincent2005evolutionary,sandholm2010population,
zhuTembineBasar2011evolutionaryMAC,hayelZhu2015evolutionaryPoisson,
liuZhaoZhu2021herd}. The Wardrop-equilibrium component is also close to
traffic and security models that study adversarial perturbations and
non-equilibrium learning in congestion networks
\cite{panZhu2022poisonedWardrop,panLiZhu2022resilienceTraffic}. The mean-field
derivation is aligned with classical convergence ideas for Markov
processes \cite{ethier2009markov}.

\noindent\textbf{Organization.}
Section~\ref{sec:continuum_ad} develops the continuum Aumann-Shapley and
Aumann-Dr\`eze values. Section~\ref{sec:continuum_exit_join} derives the
exit-and-join dynamics and the equilibrium-stationarity connection.
Section~\ref{sec:wardrop_equivalence} relates the model to Wardrop
equilibrium and variational inequalities. Section~\ref{sec:acceptance_rules}
introduces switching costs and acceptance constraints.
Section~\ref{sec:cooperation_noncooperation} discusses the relationship
between cooperative value creation and noncooperative mobility.
Section~\ref{sec:evolutionary_interpretation} gives an evolutionary
interpretation, and Section~\ref{sec:conclusion} concludes the paper.

\section{Continuum Extension of the Shapley and Aumann-Dr\`eze Values}
\label{sec:continuum_ad}

\subsection{Nonatomic Cooperative Games and the Aumann-Shapley Value}
\label{sec:nonatomic_aumann_shapley}

Let $(I,\mathcal I,\mu)$ be a nonatomic probability space, whose elements index a
continuum of agents. A \emph{nonatomic transferable-utility cooperative game} is a
set function $v:\mathcal I\to\mathbb R$ with $v(\varnothing)=0$,
which assigns a real value to each measurable coalition. Coalitions are identified
with measurable subsets of $I$, and all equalities are understood up to
$\mu$-null sets.

Because the player space is nonatomic, individual agents have zero measure and
cannot affect coalition values on their own. As a result, discrete constructions
based on permutations of players,  central to the classical Shapley value, are no
longer meaningful. Instead, solution concepts must be formulated in terms of
infinitesimal marginal contributions of population mass. Throughout this section, we assume that the game $v$ satisfies sufficient
regularity conditions to ensure the existence of well-defined marginal
contribution densities.
In finite cooperative games, the Shapley value assigns payoffs by averaging a
player’s marginal contribution over all permutations of the player set. Each
permutation represents a possible order of coalition formation, and the Shapley
value measures the expected incremental value created when a player joins the
coalition formed by its predecessors.

In a continuum of players, however, individual agents are nonatomic and have zero
measure. No agent can be meaningfully inserted at a specific position in an
ordering, and the permutation-based definition of the Shapley value does not
apply. Nevertheless, the underlying Shapley principle remains meaningful: payoffs
should reflect average marginal contributions to coalition value.

The key insight of Aumann and Shapley is that, in a nonatomic setting, discrete
insertions can be replaced by a \emph{continuous participation process}. Rather
than considering permutations of players, one considers a path along which the
measure of participating agents grows continuously from zero to full
participation. Marginal contributions are then defined infinitesimally along this
path.

\subsection{From Participation Paths to the Aumann-Shapley Value}

In finite cooperative games, the Shapley value assigns to each player
the average of its marginal contributions over all permutations of the
player set. Each permutation represents a possible order of coalition
formation, and a player’s payoff reflects the incremental value created
when it joins the coalition formed by its predecessors.

In a nonatomic game this logic no longer applies. Individual agents
have zero measure, $\mu(\{i\})=0$ for all $i\in I$, and therefore
cannot generate discrete jumps in coalition value.
There is no meaningful notion of inserting a single agent at a specific
position in an ordering. The permutation-based construction must be
replaced by a continuous analogue.

The key idea of Aumann and Shapley is to reinterpret coalition
formation as a continuous participation process. Instead of averaging
over permutations, we average marginal productivity over participation
levels.

\begin{definition}\label{def:nonatomic}
Let $(I,\mathcal I,\mu)$ be a nonatomic probability space with
$\mu(I)=1$.
A nonatomic transferable–utility cooperative game \cite{aumann2015values}
is a set function
$v:\mathcal I\to\mathbb R$ with $v(\varnothing)=0$, where coalitions
are measurable subsets of $I$.
Two coalitions that differ on a $\mu$–null set are identified.
\end{definition}

The nonatomic assumption means that value cannot change because of
the addition of a single agent. Coalition growth must therefore be
modeled at the level of population mass.

\begin{definition}\label{def:ppath}
Let $S \in \mathcal I$.
A participation path for $S$ is a measurable family
$\{S_\lambda\}_{\lambda\in[0,1]}$
such that
\[
S_0=\varnothing,
\qquad
S_1=S,
\qquad
S_\lambda \subset S_{\lambda'}
\text{ whenever } \lambda<\lambda',
\]
and
$\mu(S_\lambda)=\lambda\,\mu(S)$.
\end{definition}

\begin{example}
Let $I=[0,1]$ equipped with Lebesgue measure and let
$S \subset [0,1]$ be a measurable set with $\mu(S)=\alpha$.
Define, for $\lambda\in[0,1]$,
$S_\lambda=\{\,i\in S:i\le\lambda\alpha\,\}$.
Then $S_0=\varnothing$ and $S_1=S$. Moreover,
if $\lambda<\lambda'$, then $S_\lambda \subset S_{\lambda'}$,
and $\mu(S_\lambda)=\lambda\,\mu(S)=\lambda\alpha$.
Hence $\{S_\lambda\}_{\lambda\in[0,1]}$ is a participation path for $S$.
\end{example}

A participation path represents the gradual growth of a coalition
from empty ($\lambda=0$) to full size ($\lambda=1$).
The parameter $\lambda$ measures the aggregate level of participation.
This replaces the discrete insertion of players in the finite case.

Along such a path, coalition value evolves continuously.

 \begin{definition}\label{def:pvalue}
Let $(I,\mathcal I,\mu)$ be a nonatomic probability space,
let $v:\mathcal I\to\mathbb R$ be a cooperative game,
and let $\{S_\lambda\}_{\lambda\in[0,1]}$ be a participation path
for a measurable coalition $S$.

Assume that the map
\(
\lambda \longmapsto v(S_\lambda)
\)
is absolutely continuous on $[0,1]$.
Then its derivative exists for almost every $\lambda\in[0,1]$,
and the function
\(
\lambda \longmapsto 
\frac{d}{d\lambda} v(S_\lambda)
\)
is called the \emph{pathwise marginal value} of coalition growth
at participation level $\lambda$.
\end{definition}

The derivative above measures the instantaneous rate at which value
is created when coalition size increases infinitesimally.
It plays the role of the discrete marginal contribution
$v(S\cup\{i\})-v(S)$ in finite games.

\begin{example}
Let $(I,\mathcal I,\mu)$ be a nonatomic probability space
with $\mu(I)=1$ and suppose the game is mass-based:
\[
v(S)=F(\mu(S)),
\qquad F\in C^1([0,1]).
\]

Let $\{S_\lambda\}$ be any participation path for $S$.
Since $\mu(S_\lambda)=\lambda\,\mu(S)$, we have
\(
v(S_\lambda)
=
F(\lambda\,\mu(S)).
\)

By the chain rule,
\[
\frac{d}{d\lambda} v(S_\lambda)
=
F'(\lambda\,\mu(S))\,\mu(S).
\]

Thus the pathwise marginal value at level $\lambda$
equals the marginal productivity of coalition size
evaluated at the current coalition mass.
\end{example}

To attribute this aggregate marginal value to individual agents,
we require that it admits an integral decomposition.

\begin{definition}\label{def:mdensity}
The game $v$ admits a marginal contribution density if there exists a
measurable function
\[
m : I \times [0,1] \to \mathbb R
\]
such that for every measurable coalition $S$ and every participation path
$\{S_\lambda\}$,
\[
\frac{d}{d\lambda} v(S_\lambda)
=
\int_{S} m(i,\lambda)\, d\mu(i)
\quad
\text{for a.e. } \lambda.
\]
\end{definition}

The function $m(i,\lambda)$ represents the marginal productivity of
agent $i$ when the aggregate participation level is $\lambda$.
Thus, $\lambda$ captures coalition scale;  $m(i,\lambda)$ captures individual productivity at that scale. This is the continuum analogue of assigning discrete marginal
increments to individual players.

\begin{example}
Continuing the mass–based example, we now identify a marginal
contribution density.

Since
\[
\frac{d}{d\lambda} v(S_\lambda)
=
F'(\lambda\,\mu(S))\,\mu(S),
\]
and since $\mu(S_\lambda)=\lambda\,\mu(S)$,
we may rewrite this as
\[
\frac{d}{d\lambda} v(S_\lambda)
=
F'(\mu(S_\lambda))\,\mu(S).
\]

Because $\mu(S)=\int_S 1\, d\mu(i)$, we obtain
\[
\frac{d}{d\lambda} v(S_\lambda)
=
\int_S F'(\mu(S_\lambda))\, d\mu(i).
\]

Thus the function
\[
m(i,\lambda)
=
F'(\mu(S_\lambda))
\]
is a marginal contribution density. In particular, the density does not depend on the identity
of $i$ but only on the current coalition mass.
Hence the game admits a marginal contribution density,
and it is uniform across agents.
\end{example}

In the finite Shapley value, payoffs are obtained by averaging
marginal contributions over all permutations. In the continuum,
``averaging over permutations'' becomes
``averaging over participation levels.''

\begin{definition}\label{def:pprinciple}
Given a marginal contribution density $m(i,\lambda)$,
the average marginal contribution of agent $i$ across all participation
levels is defined as
\[
\int_0^1 m(i,\lambda)\, d\lambda.
\]
\end{definition}

This averaging step preserves the core Shapley principle:
payoffs reflect average marginal productivity across coalition scales.

\begin{definition}\label{def:asvalue}
Let $(I,\mathcal I,\mu)$ be a nonatomic probability space with $\mu(I)=1$,
and let $v:\mathcal I\to\mathbb R$ be a cooperative game admitting a
marginal contribution density
\[
m:I\times[0,1]\to\mathbb R
\]
in the sense that for every measurable coalition $S$ and every
participation path $\{S_\lambda\}$,
\[
\frac{d}{d\lambda} v(S_\lambda)
=
\int_S m(i,\lambda)\, d\mu(i)
\quad
\text{for a.e. }\lambda\in[0,1].
\]

The \emph{Aumann-Shapley value} of $v$ is the function
\[
\phi^{AS}(v) \in L^1(I,\mu)
\]
defined $\mu$–almost everywhere by
\[
\phi^{AS}_i(v)
=
\int_0^1 m(i,\lambda)\, d\lambda.
\]
\end{definition}

\begin{theorem}\label{thm:efficiency}
Suppose $v$ admits a marginal contribution density $m$.
Then the Aumann-Shapley value satisfies
\[
\int_I \phi^{AS}_i(v)\, d\mu(i)
=
v(I).
\]
\end{theorem}

\begin{proof}
Let $\{I_\lambda\}_{\lambda\in[0,1]}$ be a participation path
for the grand coalition $I$.
Since $\mu(I_\lambda)=\lambda$, we have
\[
v(I_1)-v(I_0)
=
\int_0^1
\frac{d}{d\lambda} v(I_\lambda)\, d\lambda.
\]

By the defining property of the marginal contribution density,
\[
\frac{d}{d\lambda} v(I_\lambda)
=
\int_I m(i,\lambda)\, d\mu(i).
\]

Hence,
\[
v(I)
=
\int_0^1
\int_I m(i,\lambda)\, d\mu(i)\, d\lambda.
\]

By Fubini’s theorem,
\[
v(I)
=
\int_I
\left(
\int_0^1 m(i,\lambda)\, d\lambda
\right)
d\mu(i)
=
\int_I \phi^{AS}_i(v)\, d\mu(i).
\]
\end{proof}

\begin{proposition}\label{prop:mass-based}
Suppose that
\[
v(S)=F(\mu(S)),
\qquad
F\in C^1([0,1]).
\]
Then the game admits a marginal contribution density given by
\[
m(i,\lambda)=F'(\lambda)
\quad
\text{for all } i\in I,
\]
and consequently
\[
\phi^{AS}_i(v)
=
\int_0^1 F'(\lambda)\, d\lambda
=
F(1)-F(0),
\qquad
\mu\text{–a.e. } i.
\]
\end{proposition}

\begin{proof}
Let $\{S_\lambda\}$ be a participation path for $S$.
Since $\mu(S_\lambda)=\lambda\mu(S)$,
\[
v(S_\lambda)=F(\lambda\mu(S)).
\]
Differentiating,
\[
\frac{d}{d\lambda} v(S_\lambda)
=
F'(\lambda\mu(S))\,\mu(S).
\]

For the grand coalition $I$, where $\mu(I)=1$,
this reduces to
\[
\frac{d}{d\lambda} v(I_\lambda)
=
F'(\lambda).
\]

Because
\[
F'(\lambda)
=
\int_I F'(\lambda)\, d\mu(i),
\]
the function $m(i,\lambda)=F'(\lambda)$
is a marginal contribution density.
The Aumann-Shapley value therefore satisfies
\[
\phi^{AS}_i(v)
=
\int_0^1 F'(\lambda)\, d\lambda
=
F(1)-F(0).
\]
\end{proof}

\subsection{Axiomatic Characterization}

As in the finite-player case, the Aumann-Shapley value is uniquely
characterized by natural axioms extending the classical Shapley
principle to nonatomic cooperative games.

Let $\mathcal V$ denote the class of cooperative games on
$(I,\mathcal I,\mu)$ admitting a marginal contribution density.
An allocation rule is a mapping
\(
\Phi : \mathcal V \to L^1(I,\mu)
\)
that assigns to each game $v$ a payoff density $\Phi(v)$.

\begin{theorem}\label{thm:axiomatic}
The Aumann-Shapley value is the unique allocation rule
$\Phi : \mathcal V \to L^1(I,\mu)$, defined up to $\mu$-null sets,
satisfying the following properties:
\begin{enumerate}
\item \emph{Efficiency.}
\[
\int_I \Phi_i(v)\, d\mu(i)
=
v(I).
\]

\item \emph{Symmetry.}
If two agents $i,j\in I$ satisfy
\[
m(i,\lambda)=m(j,\lambda)
\quad \text{for a.e. }\lambda\in[0,1],
\]
then
\[
\Phi_i(v)=\Phi_j(v)
\quad \text{for }\mu\text{-a.e. } i,j.
\]

\item \emph{Linearity.}
For all $v,w\in\mathcal V$ and $\alpha,\beta\in\mathbb R$,
\[
\Phi(\alpha v+\beta w)
=
\alpha \Phi(v)+\beta \Phi(w).
\]

\item \emph{Marginality.}
If two games $v,w\in\mathcal V$ admit marginal contribution densities
$m_v$ and $m_w$ satisfying
\[
m_v(i,\lambda)=m_w(i,\lambda)
\quad \text{for $\mu$-a.e. } i
\text{ and a.e. }\lambda,
\]
then
\[
\Phi(v)=\Phi(w)
\quad \text{in } L^1(I,\mu).
\]
\end{enumerate}

Under these axioms,
\[
\Phi(v)=\phi^{AS}(v).
\]
\end{theorem}

The axioms mirror the classical Shapley axioms but are formulated in
terms of marginal contribution densities rather than discrete marginal
increments. Efficiency ensures that the grand coalition value is fully
distributed. Symmetry guarantees equal treatment of agents with identical
marginal productivity functions. Linearity preserves additivity across
games. Marginality ensures that payoffs depend only on infinitesimal
marginal contributions and not on absolute value levels.

\begin{remark}
If the cooperative game $v$ is convex, the Aumann-Shapley value
coincides with the competitive (Walrasian) payoff allocation in the
associated exchange or production economy.

In this case, there exists a price functional $p$ such that each agent's
payoff equals its marginal productivity evaluated at equilibrium prices.
The grand coalition outcome can therefore be decentralized through
price-taking behavior. This establishes a structural equivalence between
the Shapley principle, marginal productivity pricing, and general
equilibrium theory in large economies.
\end{remark}

\begin{example}
\label{ex:aumann_shapley_continuum}

Let $(I,\mathcal I,\mu)$ be a nonatomic probability space with
$\mu(I)=1$.
Suppose each agent supplies one unit of a homogeneous input and total
output is generated according to
\[
F(x)=x^2,
\qquad x\in[0,1].
\]

Define the cooperative game
\[
v(S)=F(\mu(S))=\mu(S)^2.
\]

Since $F$ is convex, the game is convex.

The marginal productivity at participation level $x$ is
\[
F'(x)=2x.
\]

Hence the Aumann-Shapley value assigns
\[
\phi^{AS}_i(v)
=
\int_0^1 2x\,dx
=
1
\quad \text{for }\mu\text{-a.e. } i.
\]

Efficiency holds:
\[
\int_I \phi^{AS}_i(v)\, d\mu(i)
=
1
=
v(I).
\]

Thus each agent receives the same payoff density, reflecting both
symmetry and the price-taking nature of large competitive economies.
\end{example}

\subsection{Non-Mass-Based Cooperative Games and Kernel Interactions}
\label{subsec:functional_kernel}

The mass--based specification of Proposition~\ref{prop:mass-based}
is a special case in which coalition value depends only on its measure.
We now extend Definitions~\ref{def:ppath}, \ref{def:pvalue},
and \ref{def:mdensity} to cooperative games whose value depends
on the composition of agents.

\subsubsection{Finite-Dimensional Functional Games}

Let $(I,\mathcal I,\mu)$ be a nonatomic probability space with $\mu(I)=1$.
Let
\(
\psi_1,\dots,\psi_k \in L^1(I,\mu)
\)
be measurable characteristics (types). For every measurable coalition $S\in\mathcal I$, define
\(
T_r(S)
:=
\int_S \psi_r(i)\, d\mu(i),
\  r=1,\dots,k.
\) Let
\(
F:\mathbb R^k \to \mathbb R
\) 
be continuously differentiable.
Define the cooperative game
\(
v(S)
=
F\big(T_1(S),\dots,T_k(S)\big).
\)

This specification allows coalition value to depend on composition
rather than only on mass.
If $\psi_1\equiv 1$ and $k=1$, the model reduces to the mass--based case.

\begin{proposition}
\label{prop:functional_mdensity}
Let $v$ be defined as above and let
$\{S_\lambda\}_{\lambda\in[0,1]}$
be a participation path for $S$
as in Definition~\ref{def:ppath}.
Assume that $\lambda \mapsto S_\lambda$ is absolutely continuous
and that
\[
\frac{d}{d\lambda}\mu_{S_\lambda}(A)
=
\mu(A\cap S)
\quad \text{for all measurable } A,
\]
i.e., coalition growth is proportional.

Define
\[
T_r(\lambda)
:=
\int_{S_\lambda} \psi_r(i)\, d\mu(i),
\qquad r=1,\dots,k.
\]

Then $v$ admits a marginal contribution density
$m:I\times[0,1]\to\mathbb R$ given by
\[
m(i,\lambda)
=
\sum_{r=1}^k
\frac{\partial F}{\partial x_r}
\big(T_1(\lambda),\dots,T_k(\lambda)\big)
\,\psi_r(i).
\]
\end{proposition}

\begin{proof}
Write
\(
T(\lambda)
=
\big(T_1(\lambda),\dots,T_k(\lambda)\big)
\in \mathbb R^k.
\) Then
\(
v(S_\lambda)
=
F\big(T(\lambda)\big).
\) Since $F\in C^1(\mathbb R^k)$,
the chain rule yields
\[
\frac{d}{d\lambda} v(S_\lambda)
=
\nabla F\big(T(\lambda)\big)
\cdot
T'(\lambda)
=
\sum_{r=1}^k
\frac{\partial F}{\partial x_r}
\big(T(\lambda)\big)
\frac{d}{d\lambda}T_r(\lambda).
\]

By absolute continuity of the participation path,
\[
\frac{d}{d\lambda}T_r(\lambda)
=
\int_S \psi_r(i)\, d\mu(i).
\]

Hence
\[
\frac{d}{d\lambda} v(S_\lambda)
=
\int_S
\left(
\sum_{r=1}^k
\frac{\partial F}{\partial x_r}
\big(T(\lambda)\big)
\psi_r(i)
\right)
d\mu(i).
\]

Comparing with Definition~\ref{def:mdensity}
establishes the marginal contribution density.
\end{proof}

\begin{corollary}
\label{cor:functional_AS}
Under the above assumptions, the Aumann--Shapley value exists and satisfies
\[
\phi^{AS}_i(v)
=
\int_0^1 m(i,\lambda)\, d\lambda
=
\sum_{r=1}^k
\psi_r(i)
\int_0^1
\partial_r F\big(T(\lambda)\big)
\, d\lambda.
\]
\end{corollary}

Thus the Aumann--Shapley value is a linear combination of the
agent's characteristics $\psi_r(i)$,
with coefficients determined by averaged marginal productivity
along the participation path.

\begin{remark}
Unless all $\psi_r$ are constant $\mu$--almost everywhere,
the Aumann--Shapley value is heterogeneous across agents.
Equality of payoffs arises only under symmetry of the
marginal contribution density
(cf.~Theorem~\ref{thm:axiomatic}).
\end{remark}

\begin{example}
\label{ex:two_type_quadratic}

Let $(I,\mathcal I,\mu)$ be a nonatomic probability space with $\mu(I)=1$.
Suppose each agent has two measurable characteristics
$\psi_1(i)=\theta(i)$ and $\psi_2(i)=1$,
where $\theta \in L^1(I,\mu)$ represents productivity.

For any coalition $S$, define
$T_1(S)=\int_S \theta(i)\, d\mu(i)$ and
$T_2(S)=\mu(S)$.
Let $F(x_1,x_2)=x_1+\beta x_1 x_2$ with $\beta>0$.
Then
\[
v(S)
=
T_1(S)+\beta T_1(S)T_2(S)
=
\int_S \theta(i)\, d\mu(i)
+
\beta\!\left(\int_S \theta(i)\, d\mu(i)\right)\mu(S).
\]
Coalition value therefore depends on total productivity and its
interaction with coalition size. The model reduces to a mass--based
game only if $\theta$ is constant almost everywhere.

Let $\{S_\lambda\}$ be a proportional participation path.
Then $T_1(\lambda)=\int_{S_\lambda}\theta(i)\,d\mu(i)$
and $T_2(\lambda)=\lambda\mu(S)$.
Since
$\frac{\partial F}{\partial x_1}=1+\beta x_2$
and
$\frac{\partial F}{\partial x_2}=\beta x_1$,
Proposition~\ref{prop:functional_mdensity} gives
\[
m(i,\lambda)
=
(1+\beta T_2(\lambda))\theta(i)
+
\beta T_1(\lambda).
\]
Substituting $T_2(\lambda)=\lambda\mu(S)$,
\[
m(i,\lambda)
=
\theta(i)
+
\beta \lambda \mu(S)\theta(i)
+
\beta T_1(\lambda).
\]

Integrating over $\lambda$,
\[
\phi^{AS}_i(v)
=
\int_0^1 m(i,\lambda)\, d\lambda
=
\theta(i)
+
\beta\theta(i)\!\int_0^1 \lambda \mu(S)\, d\lambda
+
\beta\!\int_0^1 T_1(\lambda)\, d\lambda.
\]
Since $\int_0^1 \lambda\, d\lambda=\tfrac12$, we obtain
\[
\phi^{AS}_i(v)
=
\theta(i)
+
\frac{\beta}{2}\mu(S)\theta(i)
+
\beta\!\int_0^1 T_1(\lambda)\, d\lambda.
\]

The Aumann--Shapley value decomposes into a direct productivity term
$\theta(i)$, an interaction amplification term
$\frac{\beta}{2}\mu(S)\theta(i)$,
and a common surplus term
$\beta\int_0^1 T_1(\lambda)\, d\lambda$ independent of $i$.
Unless $\theta$ is constant almost everywhere,
$\phi^{AS}_i(v)\neq \phi^{AS}_j(v)$, so the allocation is heterogeneous.
\end{example}

\subsubsection{Infinite-Dimensional Functional Formulation}

The finite-dimensional construction extends naturally to fully general
measure-dependent cooperative games. We first illustrate the idea with
a kernel interaction example and then present the general formulation.

\begin{example}

Let $(I,\mathcal I,\mu)$ be a nonatomic probability space with $\mu(I)=1$.
Let $K:I\times I\to\mathbb R$ be a symmetric measurable kernel with
$K\in L^1(I\times I)$.
Define the cooperative game
\[
v(S)
=
\iint_{S\times S}
K(i,j)\, d\mu(i)\, d\mu(j).
\]

This specification captures pairwise complementarities among agents.
Coalition value depends on the entire distribution of agents in $S$,
not merely on its mass.

Let $\{S_\lambda\}_{\lambda\in[0,1]}$ be a proportional participation path
for $S$. Then
\[
v(S_\lambda)
=
\iint_{S_\lambda\times S_\lambda}
K(i,j)\, d\mu(i)\, d\mu(j).
\]

Differentiating with respect to $\lambda$ and using symmetry of $K$ yields
\[
\frac{d}{d\lambda} v(S_\lambda)
=
2
\int_{S_\lambda}
\left(
\int_{S_\lambda} K(i,j)\, d\mu(j)
\right)
d\!\left(\frac{d\mu_{S_\lambda}}{d\lambda}\right)(i).
\]

Under proportional growth,
$\frac{d}{d\lambda}\mu_{S_\lambda}(A)=\mu(A\cap S)$,
so
\[
\frac{d}{d\lambda} v(S_\lambda)
=
\int_S
2\int_{S_\lambda}
K(i,j)\, d\mu(j)\, d\mu(i).
\]

Comparing with Definition~\ref{def:mdensity},
the marginal contribution density is
\[
m(i,\lambda)
=
2
\int_{S_\lambda}
K(i,j)\, d\mu(j).
\]

By Definition~\ref{def:asvalue},
\[
\phi^{AS}_i(v)
=
\int_0^1
2
\int_{S_\lambda}
K(i,j)\, d\mu(j)
\, d\lambda.
\]

Thus each agent's payoff equals its average interaction intensity
with the coalition along the participation path.

If $K(i,j)$ depends on $i$, then generally
\[
\phi^{AS}_i(v)\neq \phi^{AS}_j(v).
\]
The mass--based model arises as the special case $K(i,j)\equiv c$,
for which $v(S)=c\,\mu(S)^2$ and the marginal contribution density
is uniform across agents.
\end{example}

\medskip

We now abstract this structure.

\begin{definition}
\label{def:measure_functional_game}

Let $\mathcal M(I)$ denote the space of finite signed measures on $I$.
Elements of $\mathcal M(I)$ describe distributions of agents
and allow coalition value to depend on the entire population profile,
not merely on its total mass.

For each coalition $S\in\mathcal I$, define the restricted measure
\(
\mu_S(A)=\mu(A\cap S),
\  A\in\mathcal I.
\)
Thus $\mu_S$ encodes the full distribution of agents inside $S$. Let
\(\mathcal V:\mathcal M(I)\to\mathbb R
\)
be Gateaux differentiable.
We define the cooperative game
\(
v(S)=\mathcal V(\mu_S).
\)

Hence coalition value depends on the measure $\mu_S$,
that is, on the distribution of agents in $S$.
\end{definition}

\medskip

Gateaux differentiability means that for every measure $\mu$
and every signed measure perturbation $\nu$,
the directional derivative
\[
D\mathcal V(\mu)(\nu)
:=
\lim_{\varepsilon\to 0}
\frac{\mathcal V(\mu+\varepsilon \nu)-\mathcal V(\mu)}{\varepsilon}
\]
exists and is linear in $\nu$.
By the Riesz representation principle for measures,
there exists a measurable function
\[
\frac{\delta \mathcal V}{\delta \mu}(\mu): I\to\mathbb R
\]
such that
\[
D\mathcal V(\mu)(\nu)
=
\int_I
\frac{\delta \mathcal V}{\delta \mu}(\mu)(i)
\, d\nu(i).
\]

The function $\frac{\delta \mathcal V}{\delta \mu}(\mu)(i)$
represents the marginal productivity of an infinitesimal
increase of mass at agent $i$ when the coalition measure is $\mu$.

\begin{proposition}
\label{prop:functional_mdensity_general}
Let $\{S_\lambda\}$ be a proportional participation path,
so that $\frac{d}{d\lambda}\mu_{S_\lambda}(A)=\mu(A\cap S)$.
If $\mathcal V$ is Gateaux differentiable at $\mu_{S_\lambda}$,
then the game admits a marginal contribution density given by
\[
m(i,\lambda)
=
\frac{\delta \mathcal V}{\delta \mu}
\big(\mu_{S_\lambda}\big)(i).
\]
\end{proposition}

\begin{proof}
Along the participation path,
\(
v(S_\lambda)
=
\mathcal V(\mu_{S_\lambda}).
\)
By the chain rule in Banach spaces,
\[
\frac{d}{d\lambda} v(S_\lambda)
=
D\mathcal V(\mu_{S_\lambda})
\!\left(
\frac{d}{d\lambda}\mu_{S_\lambda}
\right).
\]

Using the representation of the Gateaux derivative,
\[
\frac{d}{d\lambda} v(S_\lambda)
=
\int_I
\frac{\delta \mathcal V}{\delta \mu}
\big(\mu_{S_\lambda}\big)(i)
\, d\!\left(\frac{d}{d\lambda}\mu_{S_\lambda}\right)(i).
\]

Under proportional growth,
$\frac{d}{d\lambda}\mu_{S_\lambda}(A)=\mu(A\cap S)$,
so the derivative simplifies to
\[
\frac{d}{d\lambda} v(S_\lambda)
=
\int_S
\frac{\delta \mathcal V}{\delta \mu}
\big(\mu_{S_\lambda}\big)(i)
\, d\mu(i).
\]

Comparing with Definition~\ref{def:mdensity},
the marginal contribution density is
\[
m(i,\lambda)
=
\frac{\delta \mathcal V}{\delta \mu}
\big(\mu_{S_\lambda}\big)(i).
\]
\end{proof}

\begin{corollary}
Under the above assumptions,
the Aumann--Shapley value equals
\[
\phi^{AS}_i(v)
=
\int_0^1
\frac{\delta \mathcal V}{\delta \mu}
\big(\mu_{S_\lambda}\big)(i)
\, d\lambda.
\]
\end{corollary}

\medskip

Thus the Aumann--Shapley allocation equals the average
marginal productivity of agent $i$ along the participation path.
Heterogeneity arises whenever the functional derivative
depends on $i$.

\subsection{Coalition Structures in the Continuum}
\label{subsec:coalition_structure_continuum}

At any given time, the population of agents is assumed to be fully
organized into a finite number of coalitions. Coalitions do not overlap,
and no agent remains unaffiliated. Formally, a coalition structure is a
measurable partition of the entire player space.

\begin{definition}
Let $(I,\mathcal I,\mu)$ be a nonatomic probability space.
A \emph{coalition structure} is a finite measurable partition
\[
T = \{T_1,\dots,T_m\},
\qquad m<\infty,
\]
of $I$ such that:
\[
\mu(T_i \cap T_j)=0 \quad \text{for } i\neq j,
\]
and
\[
\mu\!\left(I \setminus \bigcup_{j=1}^m T_j\right)=0.
\]
\end{definition}

The partition condition ensures that every agent belongs to exactly one
coalition, up to $\mu$-null sets. Since the player space is nonatomic,
coalition structures that differ only on null sets are identified.
Thus coalition membership is defined almost everywhere.

At a fixed time, the coalition structure is treated as exogenously
given, reflecting organizational, institutional, or technological
constraints that limit feasible coalitions. Dynamic reconfiguration
occurs through exit–and–join moves, but at every instant the coalitions
jointly exhaust the population.

Fix a coalition $T_j \in T$.
The restriction of the player space to $T_j$ is the measure space
\(
(T_j,\mathcal I_{T_j},\mu_{T_j}),
\)
where
\[
\mathcal I_{T_j}
=
\{ A \subset T_j : A = S \cap T_j
\text{ for some } S\in\mathcal I \},
\]
and
\[
\mu_{T_j}(A)
=
\mu(A),
\qquad
A \in \mathcal I_{T_j}.
\]

Since $(I,\mathcal I,\mu)$ is nonatomic, each restricted measure space
$(T_j,\mathcal I_{T_j},\mu_{T_j})$ is also nonatomic.
Thus each coalition inherits the same measure-theoretic structure as
the entire population.

Let $v:\mathcal I\to\mathbb R$ be a cooperative game.
The \emph{restricted game} on coalition $T_j$ is the function
\(
v_{|T_j} : \mathcal I_{T_j} \to \mathbb R
\)
defined by
\[
v_{|T_j}(A)
=
v(A),
\qquad
A \in \mathcal I_{T_j}.
\]

Equivalently, for any $S\in\mathcal I$,
\(
v_{|T_j}(S\cap T_j)
=
v(S\cap T_j).
\)

The restricted game captures the value generated by subcoalitions
formed entirely within $T_j$, abstracting from interactions with agents
outside the coalition. In other words, each coalition in a structure
$T$ induces its own nonatomic cooperative game, to which solution
concepts such as the Aumann-Shapley value may be applied independently.

This restriction principle is fundamental for extending the
Aumann-Dr\`eze construction to the continuum: the coalition structure
is treated as fixed, the original game is restricted to each coalition,
and payoffs are computed separately within each restricted measure space.

\begin{example}
Let $I=[0,1]$ equipped with Lebesgue measure and fix
$\alpha\in(0,1)$. Define $T_1=[0,\alpha]$ and
$T_2=(\alpha,1]$. Then $T=\{T_1,T_2\}$ is a coalition
structure. Indeed, both $T_1$ and $T_2$ are measurable,
$\mu(T_1\cap T_2)=0$, and
$\mu\!\left(I\setminus(T_1\cup T_2)\right)=0$.
The coalition masses are $\mu(T_1)=\alpha$
and $\mu(T_2)=1-\alpha$.

The restricted measure spaces are
$(T_1,\mathcal I_{T_1},\mu_{T_1})$
and $(T_2,\mathcal I_{T_2},\mu_{T_2})$,
each of which remains nonatomic since the original
space $(I,\mathcal I,\mu)$ is nonatomic. If the cooperative game is mass-based, that is,
$v(S)=F(\mu(S))$, then the restricted games satisfy
$v_{|T_1}(A)=F(\mu(A))$ for $A\subset T_1$,
and similarly $v_{|T_2}(A)=F(\mu(A))$ for
$A\subset T_2$.

Thus each coalition behaves as an independent
nonatomic cooperative game with total mass
$\alpha$ and $1-\alpha$, respectively.
\end{example}

 \begin{remark}
\label{rem:coalition_as_type}

A coalition structure can be embedded into the finite–dimensional
functional game framework by introducing indicator characteristics
\[
\psi_j(i)=\mathbf 1_{T_j}(i),
\qquad j=1,\dots,m.
\]

For any measurable coalition $S\subset I$,
\[
T_j(S)
=
\int_S \psi_j(i)\,d\mu(i)
=
\mu(S\cap T_j).
\]

Thus a cooperative game of the form
\[
v(S)
=
F\big(T_1(S),\dots,T_m(S)\big)
=
F\big(\mu(S\cap T_1),\dots,\mu(S\cap T_m)\big)
\]
depends only on the mass of $S$ inside each block of the partition.
In this sense, coalition membership functions as a finite type system,
and the coalition structure is a special case of the
finite–dimensional functional game.

Under this representation, agents within the same block $T_j$
are symmetric, and heterogeneity arises only across blocks.
The mass–based model corresponds to the degenerate case $m=1$.
\end{remark}

\subsection{Continuum Aumann-Dr\`eze Value}
\label{subsec:continuum_ad_value}

We now extend the Aumann-Dr\`eze construction (see \cite{aumann1974cooperative}) to a continuum of players.
The logic mirrors the finite case: the coalition structure is treated as fixed,
the original game is restricted to each coalition, and the Shapley principle
(in its continuum form) is applied within each restricted game.

Let $(I,\mathcal I,\mu)$ be a nonatomic player space with $\mu(I)=1$,
let $v:\mathcal I\to\mathbb R$ be a transferable–utility cooperative game,
and let $T=\{T_1,\dots,T_m\}$ be a coalition structure. For each coalition
$T_j\in T$, denote by $(T_j,\mathcal I_{T_j},\mu_{T_j})$ the restricted measure
space, where $\mathcal I_{T_j}=\{S\cap T_j:S\in\mathcal I\}$ and
$\mu_{T_j}(A)=\mu(A)$ for $A\in\mathcal I_{T_j}$, and define the restricted
game $v_{|T_j}:\mathcal I_{T_j}\to\mathbb R$ by $v_{|T_j}(A)=v(A)$.

\begin{definition}
Assume that for every coalition $T_j\in T$ the restricted game $v_{|T_j}$
admits an Aumann-Shapley value on $(T_j,\mathcal I_{T_j},\mu_{T_j})$.
The \emph{continuum Aumann-Dr\`eze value} is the payoff density
\[
\Omega^{AD}(v;T)=\bigl(\Omega^{AD}_i(v;T)\bigr)_{i\in I}\in L^1(I,\mu)
\]
defined $\mu$-almost everywhere by
\[
\Omega^{AD}_i(v;T)=\phi^{AS}_i\!\left(v_{|T_j}\right),
\qquad \text{for } i\in T_j,
\]
where $\phi^{AS}$ denotes the Aumann-Shapley value computed within the
restricted game $v_{|T_j}$.
\end{definition}

Thus each agent’s payoff is determined by its infinitesimal marginal contribution
within its own coalition, averaged over participation levels. Each coalition is
treated as an independent nonatomic cooperative game, and payoffs are computed
separately within each coalition.

\begin{remark}
Since $(I,\mathcal I,\mu)$ is nonatomic, $\Omega^{AD}(v;T)$ is defined only up to
$\mu$-null sets. All payoff identities and efficiency statements are therefore
understood to hold $\mu$-almost everywhere.
\end{remark}

\begin{remark}
If the player space consists of finitely many atoms with equal mass, the above
construction reduces to the classical Aumann-Dr\`eze value. The continuum
Aumann-Dr\`eze value may therefore be interpreted as the large-population limit
of the finite exit-and-join payoff rule.
\end{remark}

\begin{proposition}
For every coalition $T_j\in T$,
\[
\int_{T_j}\Omega^{AD}_i(v;T)\,d\mu(i)=v(T_j).
\]
\end{proposition}

\begin{proof}
By definition, $\Omega^{AD}_i(v;T)$ coincides $\mu$-almost everywhere on $T_j$
with $\phi^{AS}_i(v_{|T_j})$. Efficiency of the Aumann-Shapley value on the
restricted measure space implies
\[
\int_{T_j}\phi^{AS}_i(v_{|T_j})\,d\mu(i)=v_{|T_j}(T_j)=v(T_j),
\]
which yields the claim.
\end{proof}

\begin{remark}
As in the finite Aumann-Dr\`eze construction, the coalition structure is treated
as exogenously fixed. Payoffs within each coalition depend only on the restricted
game on that coalition; cross-coalition externalities are not internalized.
\end{remark}

\begin{example}
\label{ex:continuum_ad}

Let $(I,\mathcal I,\mu)$ be a nonatomic probability space with $\mu(I)=1$ and
consider the cooperative game $v(S)=\mu(S)^2$ for $S\in\mathcal I$.
Let $T=\{T_1,T_2\}$ be a coalition structure with $\mu(T_1)=\alpha$ and
$\mu(T_2)=1-\alpha$ for some $\alpha\in(0,1)$. Denote $m_1=\mu(T_1)=\alpha$ and
$m_2=\mu(T_2)=1-\alpha$.
For $j\in\{1,2\}$ the restricted game on $T_j$ is given by
$v_{|T_j}(A)=\mu(A)^2$ for $A\in\mathcal I_{T_j}$. In particular the value of the
grand coalition within $T_j$ is $v_{|T_j}(T_j)=v(T_j)=m_j^2$.
Fix $j\in\{1,2\}$ and choose any participation path $\{A_\lambda\}_{\lambda\in[0,1]}$
for $T_j$. By definition, $\mu(A_\lambda)=\lambda m_j$. Along this path,
\[
v_{|T_j}(A_\lambda)=\mu(A_\lambda)^2=(\lambda m_j)^2.
\]
Differentiating with respect to $\lambda$ yields
\[
\frac{d}{d\lambda}v_{|T_j}(A_\lambda)=2(\lambda m_j)m_j=2\lambda m_j^2.
\]
Equivalently, if we reparametrize by coalition mass $x=\mu(A_\lambda)=\lambda m_j$,
then the marginal productivity at level $x\in[0,m_j]$ is $2x$.
The Aumann-Shapley value averages marginal productivity over participation
levels. Using the mass parameterization on $[0,m_j]$ gives
\[
\phi^{AS}_i(v_{|T_j})
=
\frac{1}{m_j}\int_0^{m_j}2x\,dx
=
\frac{1}{m_j}\bigl[x^2\bigr]_{0}^{m_j}
=
m_j,
\qquad i\in T_j.
\]
By definition of $\Omega^{AD}$,
\[
\Omega^{AD}_i(v;T)=m_j \quad \text{for } i\in T_j,
\]
so
\[
\Omega^{AD}_i(v;T)=
\begin{cases}
\alpha, & i\in T_1,\\
1-\alpha, & i\in T_2.
\end{cases}
\]
We verify efficiency within each coalition:
\begin{align*}
\int_{T_1}\Omega^{AD}_i(v;T)\,d\mu(i)
&=\alpha\,\mu(T_1)=\alpha^2=v(T_1),\\
\int_{T_2}\Omega^{AD}_i(v;T)\,d\mu(i)
&=(1-\alpha)\,\mu(T_2)=(1-\alpha)^2=v(T_2).
\end{align*}
\end{example}

\begin{theorem}
\label{thm:ad_membership_connection}

Let $T=\{T_1,\dots,T_m\}$ be a coalition structure on
$(I,\mathcal I,\mu)$.
Define membership functions
\[
\psi_j(i)=\mathbf 1_{T_j}(i),
\qquad j=1,\dots,m,
\]
and consider a cooperative game of finite-dimensional functional form
\[
v(S)
=
F\big(T_1(S),\dots,T_m(S)\big),
\qquad
T_j(S)=\mu(S\cap T_j),
\]
where $F:\mathbb R^m\to\mathbb R$ is continuously differentiable.

Assume that $F$ exhibits no cross-coalition externalities in the sense that,
for every $j$,
\[
F(x_1,\dots,x_m)
=
F_j(x_j)
\quad
\text{whenever } x_k=0 \text{ for } k\neq j.
\]

Then the continuum Aumann--Dr\`eze value satisfies
\[
\Omega^{AD}_i(v;T)
=
\int_0^1
\frac{\partial F_j}{\partial x}
\big(\lambda\,\mu(T_j)\big)
\, d\lambda,
\qquad i\in T_j,
\]
and coincides $\mu$-almost everywhere with the blockwise
Aumann--Shapley value computed within each coalition $T_j$.
\end{theorem}

\begin{proof}
Under the membership representation,
\(
T_j(S)
=
\int_S \mathbf 1_{T_j}(i)\,d\mu(i)
=
\mu(S\cap T_j).
\)
 Fix a coalition $T_j$ and consider the restricted game
$v_{|T_j}$.
For any $A\subset T_j$,
\[
T_k(A)=0 \text{ for } k\neq j,
\qquad
T_j(A)=\mu(A).
\]

Hence, by the no-externality assumption,
\[
v_{|T_j}(A)
=
F(0,\dots,\mu(A),\dots,0)
=
F_j(\mu(A)).
\]

Thus the restricted game is mass-based on $T_j$.

Let $\{A_\lambda\}$ be a participation path for $T_j$.
Then $\mu(A_\lambda)=\lambda\,\mu(T_j)$ and
\[
v_{|T_j}(A_\lambda)
=
F_j(\lambda\,\mu(T_j)).
\]

Differentiating,
\[
\frac{d}{d\lambda}v_{|T_j}(A_\lambda)
=
F_j'(\lambda\,\mu(T_j))\,\mu(T_j).
\]

By the definition of the Aumann--Shapley value within $T_j$,
\[
\phi^{AS}_i(v_{|T_j})
=
\int_0^1
F_j'(\lambda\,\mu(T_j))
\, d\lambda,
\qquad i\in T_j.
\]

By Definition~\ref{subsec:continuum_ad_value},
\[
\Omega^{AD}_i(v;T)
=
\phi^{AS}_i(v_{|T_j}),
\qquad i\in T_j,
\]
which proves the claim.
\end{proof}

\section{Exit-and-Join Dynamics in the Continuum}
\label{sec:continuum_exit_join}

This section defines exit-and-join dynamics for a continuum of players in a
measure-theoretically rigorous manner. Because the player space is nonatomic,
individual agents have zero measure and cannot be assigned trajectories.
Accordingly, coalition reconfiguration is described as an evolution of
population measures, induced by payoff differences under the continuum
Aumann-Dr\`eze value.

\subsection{Primitives, State Space, and Coalition Payoffs}

We begin by specifying the primitives of the model and the induced
state representation.
Let $(I,\mathcal I,\mu)$ be a nonatomic probability space,
so that $\mu(I)=1$ and $\mu(\{i\})=0$ for all $i\in I$.
Let $v:\mathcal I\to\mathbb R$ be a transferable–utility cooperative
game defined on measurable coalitions.
We assume throughout that $v$ satisfies the regularity conditions
required for the existence of the Aumann–Shapley value on measurable
subsets and, consequently, for the continuum Aumann–Dr\`eze value
on measurable partitions.
Fix $m<\infty$.
A \emph{coalition structure} is a finite measurable partition
\[
T=\{T_1,\dots,T_m\}\subset\mathcal I
\]
satisfying
\[
\mu(T_i\cap T_j)=0 \quad \text{for } i\neq j,
\qquad
\mu\!\left(I\setminus\bigcup_{i=1}^m T_i\right)=0.
\]
Coalition structures that differ only on $\mu$–null sets are identified.
At time $t\ge 0$, the system is described by a coalition structure $T(t)$.
The associated \emph{population state} is
\[
x(t)=(x_1(t),\dots,x_m(t))\in\mathbb R_+^m,
\qquad
x_i(t):=\mu\!\left(T_i(t)\right).
\]
Since $\sum_{i=1}^m x_i(t)=1$, the state lies in the simplex
\[
\Delta_m
=
\left\{
x\in\mathbb R_+^m:
\sum_{i=1}^m x_i=1
\right\}.
\]

\begin{remark}
The mapping $T\mapsto x$ is many–to–one:
distinct coalition structures may induce the same population state.
However, coalition payoffs and exit–and–join dynamics
will depend only on the induced state $x\in\Delta_m$.
\end{remark}

For each coalition structure $T=\{T_1,\dots,T_m\}$,
let
\[
\Omega^{AD}(v;T)
=
\bigl(\Omega^{AD}_k(v;T)\bigr)_{k\in I}
\]
denote the continuum Aumann–Dr\`eze payoff density.

\begin{assumption}
\label{ass:coalition_symmetry}
For every coalition $T_i\in T$, the restricted game $v_{|T_i}$
is symmetric in the sense that its Aumann–Shapley value assigns
the same payoff density to $\mu$–almost every agent in $T_i$.
\end{assumption}

\begin{lemma}
\label{lem:coalition_payoff_level}
Under Assumption~\ref{ass:coalition_symmetry},
for each coalition $T_i$ there exists a constant
$r_i(T)\in\mathbb R$ such that
\[
\Omega^{AD}_k(v;T)=r_i(T)
\quad \text{for $\mu$–almost every } k\in T_i.
\]
\end{lemma}

\begin{proof}
By coalition symmetry,
$\Omega^{AD}_k(v;T)$ is constant $\mu$–almost everywhere
on $T_i$. The constant is uniquely defined up to
$\mu$–null sets and is denoted by $r_i(T)$.
\qedhere
\end{proof}
Let $x\in\Delta_m$ be induced by a coalition structure $T$.
Define the \emph{coalition payoff function}
\[
\rho_i:\Delta_m\to\mathbb R,
\qquad
\rho_i(x):=r_i(T),
\quad i=1,\dots,m.
\]
The corresponding payoff vector field is
\[
\rho(x)
=
(\rho_1(x),\dots,\rho_m(x)).
\]

\begin{remark}
Although the definition $\rho_i(x):=r_i(T)$ is expressed
in terms of the coalition structure $T$ inducing $x$,
the dependence on $T$ is only through the coalition mass
$\mu(T_i)=x_i$.
Under Assumption~\ref{ass:coalition_symmetry},
the Aumann–Shapley value within a coalition depends
on the restricted game solely via the coalition measure.
Hence the $x$–dependence of $\rho_i(x)$ is implicit
through $x_i$, and no additional structural features
of the partition affect the payoff level.
\end{remark}
\begin{proposition}
\label{prop:state_dependence}
Under Assumption~\ref{ass:coalition_symmetry},
the payoff vector $\rho(x)$ depends on the coalition structure
only through the induced population state $x\in\Delta_m$.
\end{proposition}

\begin{proof}
Suppose coalition structures $T$ and $T'$ induce the same state
$x\in\Delta_m$. Then $\mu(T_i)=\mu(T_i')$ for all $i$.
Under coalition symmetry,
the Aumann–Shapley value within each coalition depends only
on the restricted game and the coalition mass.
Hence $r_i(T)=r_i(T')$ for all $i$,
so $\rho_i(x)$ is well defined independently of
the specific partition realizing $x$.
\qedhere
\end{proof}

\subsection{Admissible Switching Incentives}

Because the player space $(I,\mathcal I,\mu)$ is nonatomic,
each individual agent has zero measure and therefore cannot
affect the population state on its own.
Incentives must therefore be understood at the level of
individual switching decisions, which aggregate into
movements of population mass.

Let $x\in\Delta_m$ be the population state induced by a coalition
structure $T=\{T_1,\dots,T_m\}$.
Under Assumption~\ref{ass:coalition_symmetry},
the continuum Aumann–Dr\`eze value assigns a payoff density
that is constant $\mu$–almost everywhere on each coalition.
We denote by
\(
\rho_i(x)
\)
the payoff level assigned to coalition $T_i$.
Thus, for $\mu$–almost every agent $k\in T_i$,
its payoff equals $\rho_i(x)$.

\begin{definition}
\label{def:admissible_switch}
Coalition $j$ offers an \emph{admissible switching opportunity}
to agents in coalition $i$ at state $x$
if
\[
\rho_j(x) > \rho_i(x).
\]
\end{definition}

Admissibility is evaluated at the current population state $x$.
It captures the existence of a strict payoff advantage
for agents currently in coalition $i$
to switch voluntarily to coalition $j$.

\begin{lemma}
\label{lem:positive_measure_switch}
If $\rho_j(x)>\rho_i(x)$,
then for every measurable subset $S\subseteq T_i$
with $\mu(S)>0$,
almost every agent in $S$ strictly prefers
switching to coalition $T_j$
over remaining in $T_i$.
\end{lemma}

\begin{proof}
Under coalition symmetry,
almost every agent in $T_i$ receives payoff $\rho_i(x)$.
Hence almost every agent in $S\subseteq T_i$
receives payoff $\rho_i(x)$.

If such an agent were to switch to coalition $T_j$
while the population state remains $x$,
its payoff would equal $\rho_j(x)$.
Since $\rho_j(x)>\rho_i(x)$,
switching strictly increases payoff
for almost every agent in $S$.

Because $\mu(S)>0$,
there exists a set of positive measure
whose members strictly prefer to switch.
\qedhere
\end{proof}

\begin{remark}
In the nonatomic setting,
strict payoff inequality between coalitions
generates decentralized switching incentives.
No central planner reallocates agents.
Rather, when $\rho_j(x)>\rho_i(x)$,
almost every agent in coalition $i$
individually prefers coalition $j$.
Aggregate movements of population mass
therefore arise endogenously
from individual incentives.
\end{remark}

\subsection{Finite-Population Approximation and Mean-Field Limit}
\label{subsec:derivation_mass_dynamics}

The continuum dynamics should not be interpreted as assigning an
independent Poisson clock to each individual point of the nonatomic
space. A cleaner construction is to approximate the continuum by
finite populations and then pass to a law-of-large-numbers limit.
This yields the deterministic coalition-mass equation as the limit of
density-dependent Markov chains.

For each ordered pair $(i,j)$ with $i\neq j$, let
\[
\lambda_{ij}:\Delta_m\to\mathbb R_+
\]
be the per-agent switching intensity from coalition $i$ to coalition
$j$. The intensity may depend on the current population state.

\begin{assumption}
\label{ass:mf_regularity}
For each $i\neq j$, the function $\lambda_{ij}$ is Lipschitz continuous
on $\Delta_m$. Moreover, there is a constant $\Lambda<\infty$ such that
\[
0\le \lambda_{ij}(x)\le \Lambda,
\qquad x\in\Delta_m,\quad i\neq j.
\]
\end{assumption}

The boundedness condition is automatic if the $\lambda_{ij}$ are
continuous on the compact simplex. We state it explicitly because it is
the estimate used in the martingale bound below.

For $N\in\mathbb N$, consider $N$ agents distributed among the $m$
coalitions. Let
\[
X^N(t)=\bigl(X_1^N(t),\dots,X_m^N(t)\bigr)
\in
\Delta_m^N,
\qquad
\Delta_m^N:=\{x\in\Delta_m:Nx_i\in\mathbb N\},
\]
where $X_i^N(t)$ is the fraction of agents in coalition $i$ at time
$t$. Conditional on $X^N(t)=x$, a transition from coalition $i$ to
coalition $j$ changes the empirical state by
\[
x\longmapsto x+\frac1N(e_j-e_i),
\qquad i\neq j,
\]
and occurs at aggregate rate
\begin{equation}
\label{eq:finite_rates}
q_{ij}^N(x)=N x_i\lambda_{ij}(x).
\end{equation}
Thus each of the $Nx_i$ agents currently in coalition $i$ switches to
coalition $j$ with intensity $\lambda_{ij}(x)$.

The generator $\mathcal L^N$ of $X^N(\cdot)$ acts on bounded test
functions $f:\Delta_m^N\to\mathbb R$ by
\begin{equation}
\label{eq:finite_generator}
\mathcal L^N f(x)
=
\sum_{i\neq j}
N x_i\lambda_{ij}(x)
\left[
f\!\left(x+\frac1N(e_j-e_i)\right)-f(x)
\right].
\end{equation}

Define $b:\Delta_m\to\mathbb R^m$ by
\begin{equation}
\label{eq:mass_balance}
b_i(x)
=
\sum_{j\neq i}\lambda_{ji}(x)x_j
-
\sum_{j\neq i}\lambda_{ij}(x)x_i,
\qquad i=1,\dots,m.
\end{equation}
This is the net inflow into coalition $i$ minus the outflow from
coalition $i$. Under Assumption~\ref{ass:mf_regularity}, the map $b$
is Lipschitz on $\Delta_m$ and bounded.

\begin{theorem}
\label{thm:exit_join_ode}
Assume Assumption~\ref{ass:mf_regularity}. If
$X^N(0)\to x_0\in\Delta_m$ in probability, then for every $T<\infty$,
\[
\sup_{0\le t\le T}\|X^N(t)-x(t)\|
\xrightarrow[N\to\infty]{\mathbb P}0,
\]
where $x(\cdot)$ is the unique solution of
\begin{equation}
\label{eq:mf_ode}
\dot x(t)=b(x(t)),
\qquad x(0)=x_0.
\end{equation}
Equivalently, for each $i=1,\dots,m$,
\[
\dot{x}_i(t)
=
\sum_{j\neq i}\lambda_{ji}(x(t))x_j(t)
-
\sum_{j\neq i}\lambda_{ij}(x(t))x_i(t).
\]
\end{theorem}

\begin{proof}
Since $b$ is Lipschitz on the compact simplex, it admits a Lipschitz
extension to a neighborhood of $\Delta_m$. Picard--Lindelof therefore
gives a unique local solution, and boundedness of $b$ gives global
existence. The solution remains in $\Delta_m$: indeed,
$\sum_i b_i(x)=0$, and if $x_i=0$ then
$b_i(x)=\sum_{j\neq i}\lambda_{ji}(x)x_j\ge 0$.

For each coordinate $i$, Dynkin's formula applied to the coordinate map
$f_i(x)=x_i$
gives the martingale
\begin{equation}
\label{eq:martingale}
M_i^N(t)
:=
X_i^N(t)-X_i^N(0)-\int_0^t b_i(X^N(s))\,ds .
\end{equation}
Indeed, using \eqref{eq:finite_generator} with $f_i(x)=x_i$ gives
\[
\mathcal L^N f_i(x)
=
\sum_{j\neq i}x_j\lambda_{ji}(x)
-
\sum_{j\neq i}x_i\lambda_{ij}(x)
=b_i(x).
\]

Each jump changes a single coordinate by at most $1/N$. Since the total
jump rate is bounded by
\[
\sum_{i\neq j}N x_i\lambda_{ij}(x)
\le
N(m-1)\Lambda,
\]
the predictable quadratic variation satisfies, for a constant $C$
independent of $N$,
\[
\langle M_i^N\rangle(t)\le \frac{Ct}{N}.
\]
Doob's inequality therefore gives, for each $T<\infty$,
\[
\mathbb E\!\left[
\sup_{0\le t\le T}|M_i^N(t)|^2
\right]
\le
\frac{C_T}{N}.
\]
Thus $\sup_{0\le t\le T}\|M^N(t)\|\to 0$ in probability.

Let $e^N(t)=X^N(t)-x(t)$. Combining \eqref{eq:martingale} with
\eqref{eq:mf_ode} yields
\[
e^N(t)
=
X^N(0)-x_0
+
\int_0^t\bigl(b(X^N(s))-b(x(s))\bigr)\,ds
+
M^N(t).
\]
If $L$ is a Lipschitz constant for $b$, then
\[
\sup_{0\le s\le t}\|e^N(s)\|
\le
\|X^N(0)-x_0\|
+
L\int_0^t\sup_{0\le r\le s}\|e^N(r)\|\,ds
+
\sup_{0\le s\le t}\|M^N(s)\|.
\]
Grönwall's inequality gives
\[
\sup_{0\le t\le T}\|e^N(t)\|
\le
e^{LT}
\left(
\|X^N(0)-x_0\|
+
\sup_{0\le t\le T}\|M^N(t)\|
\right).
\]
The right-hand side converges to zero in probability, proving the
uniform-on-compact convergence.
\end{proof}

\begin{remark}
The deterministic ODE \eqref{eq:mf_ode} is the continuum model. It is
not obtained by counting events in an uncountable population. Rather, it
is the hydrodynamic limit of finite empirical coalition distributions.
The nonatomic model records only the limiting coalition masses.
\end{remark}

\begin{assumption}
\label{ass:incentive_compatibility}
The switching intensities are incentive compatible if, for every
$x\in\Delta_m$ and $i\neq j$,
\[
\rho_j(x)\le \rho_i(x)
\quad\Longrightarrow\quad
\lambda_{ij}(x)=0.
\]
They are strictly payoff-responsive if, in addition,
\[
\rho_j(x)>\rho_i(x)
\quad\Longrightarrow\quad
\lambda_{ij}(x)>0.
\]
\end{assumption}

Under incentive compatibility, the vector field generates no flow from a
coalition to another coalition with weakly lower payoff. Under strict
payoff responsiveness, every strictly profitable destination generates a
positive per-agent switching rate.

\subsection{Invariance and Incentive-Compatible Dynamics}
\label{subsec:invariance_lyapunov}

\begin{proposition}
\label{prop:mass_conservation}
Let $x(t)$ evolve according to the exit–and–join dynamics
\[
\dot{x}_i(t)
=
\sum_{j\neq i}\lambda_{ji}(x(t))\,x_j(t)
-
\sum_{j\neq i}\lambda_{ij}(x(t))\,x_i(t),
\qquad i=1,\dots,m.
\]
Then the simplex
\[
\Delta_m
=
\Bigl\{ x\in\mathbb R_+^m : \sum_{i=1}^m x_i = 1 \Bigr\}
\]
is forward invariant.
\end{proposition}

\begin{proof}
Summing the right–hand side over $i$ yields
\begin{align*}
\sum_{i=1}^m \dot{x}_i(t)
&=
\sum_{i=1}^m \sum_{j\neq i}\lambda_{ji}(x)x_j
-
\sum_{i=1}^m \sum_{j\neq i}\lambda_{ij}(x)x_i.
\end{align*}
Reindexing the first double sum shows that each term
$\lambda_{ij}(x)x_i$
appears exactly once with positive sign and once with negative sign.
Hence
\[
\sum_{i=1}^m \dot{x}_i(t)=0.
\]
Therefore $\sum_i x_i(t)$ remains constant along trajectories,
and if $x(0)\in\Delta_m$ then $\sum_i x_i(t)=1$ for all $t\ge 0$.

To verify nonnegativity, suppose $x_i(t_0)=0$ for some $t_0$.
Then all outflow terms from coalition $i$ vanish since they are
proportional to $x_i(t_0)$.
The inflow terms are nonnegative.
Thus $\dot{x}_i(t_0)\ge 0$.
Therefore trajectories cannot cross the boundary of $\mathbb R_+^m$,
and $\Delta_m$ is forward invariant.
\end{proof}

\subsubsection{Replicator Dynamics as a Special Case}

\begin{proposition}
\label{prop:replicator_reduction}
Suppose the switching intensities satisfy the payoff–difference rule
\[
\lambda_{ij}(x)
=
\kappa x_j\, [\rho_j(x)-\rho_i(x)]_+,
\qquad \kappa>0,
\]
where $[a]_+=\max\{a,0\}$.
Then the exit–and–join dynamics reduce to the replicator equation
\[
\dot{x}_i
=
\kappa x_i\bigl(\rho_i(x)-\bar\rho(x)\bigr),
\qquad
\bar\rho(x)=\sum_{j=1}^m x_j\rho_j(x).
\]
\end{proposition}

\begin{proof}
Substituting the rate specification into the mass–balance system yields
\begin{align*}
\dot{x}_i
&=
\kappa x_i\sum_{j\neq i}x_j[\rho_i-\rho_j]_+
-
\kappa x_i\sum_{j\neq i}x_j[\rho_j-\rho_i]_+.
\end{align*}

Using the identity
\[
[a-b]_+ - [b-a]_+ = a-b
\quad\text{for all } a,b\in\mathbb R,
\]
we obtain
\[
x_i\sum_{j\neq i}x_j[\rho_i-\rho_j]_+
-
x_i\sum_{j\neq i}x_j[\rho_j-\rho_i]_+
=
x_i\sum_{j\neq i}x_j(\rho_i-\rho_j).
\]

Since $\sum_{j\neq i}x_j=1-x_i$, we have
\[
\sum_{j\neq i}x_j(\rho_i-\rho_j)
=
\rho_i-\sum_{j=1}^m x_j\rho_j
=
\rho_i-\bar\rho.
\]

Combining the expressions gives
\[
\dot{x}_i
=
\kappa x_i(\rho_i-\bar\rho),
\]
which is precisely the replicator equation.
\end{proof}

\subsubsection{General Incentive-Compatible Switching Rules}

The payoff difference rule leading to replicator dynamics is not unique.
More generally, the one-way incentive condition
\[
\lambda_{ij}(x)>0
\quad\Rightarrow\quad
\rho_j(x)>\rho_i(x),
\qquad i\neq j,
\]
admits a broad class of admissible switching specifications.

\begin{proposition}
\label{prop:general_phi_rule}
Let $\Phi:\mathbb R\to\mathbb R_+$ be a locally Lipschitz function satisfying
\[
\Phi(z)=0 \quad \text{for } z\le 0,
\qquad
\Phi(z)>0 \quad \text{for } z>0.
\]
Define the switching intensities by
\[
\lambda_{ij}(x)
=
x_j \, \Phi\!\big(\rho_j(x)-\rho_i(x)\big),
\qquad i\neq j.
\]
Then the induced exit–and–join dynamics are incentive compatible.
\end{proposition}

\begin{proof}
If $\rho_j(x)\le \rho_i(x)$,
then $\rho_j(x)-\rho_i(x)\le 0$,
so $\Phi(\rho_j-\rho_i)=0$ and hence
$\lambda_{ij}(x)=0$.
If $\rho_j(x)>\rho_i(x)$,
then $\Phi(\rho_j-\rho_i)>0$,
so switching from coalition $i$ to coalition $j$
occurs with positive intensity whenever the destination coalition has
positive mass, $x_j>0$. If $x_j=0$, the rule is still incentive
compatible, but the empty destination is not accessible under this
pairwise-imitation specification.
\end{proof}

Under this specification, the mass–balance system becomes
\[
\dot{x}_i
=
x_i\sum_{j\neq i} x_j
\Big[
\Phi\!\big(\rho_i(x)-\rho_j(x)\big)
-
\Phi\!\big(\rho_j(x)-\rho_i(x)\big)
\Big],
\qquad i=1,\dots,m.
\]

\subsubsection{Switching Rules and Induced Population Dynamics}

If $\Phi(z)=\kappa z_+$ with $\kappa>0$ and
$z_+=\max\{z,0\}$, then we obtain
\[
\dot{x}_i
=
\kappa x_i
\sum_{j\neq i} x_j(\rho_i-\rho_j)
=
\kappa x_i\bigl(\rho_i(x)-\bar\rho(x)\bigr),
\]
where
\[
\bar\rho(x)=\sum_{j=1}^m x_j\rho_j(x).
\]
Hence the classical replicator equation arises as the linear
payoff–difference case.

Let
\(
\Phi(z)=\kappa z_+^\alpha,
\  \alpha>0.
\)
Then the induced dynamics become
\[
\dot x_i
=
\kappa x_i
\sum_{j\neq i}
x_j
\Bigl[
(\rho_i-\rho_j)_+^\alpha
-
(\rho_j-\rho_i)_+^\alpha
\Bigr].
\]

When $\alpha=1$, this reduces to the replicator equation.
For $\alpha\neq 1$, the adjustment intensity depends nonlinearly
on payoff differences. Large payoff gaps may generate
disproportionately strong flows when $\alpha>1$,
while $\alpha<1$ dampens large differences.
The case $\alpha\ge 1$ is locally Lipschitz and fits the regularity
assumption above; the sublinear case $0<\alpha<1$ is continuous but not
Lipschitz at zero payoff gaps and requires separate well-posedness
arguments.

Let
\(
\Phi(z)=\kappa \mathbf 1_{\{z>\varepsilon\}},
\  \varepsilon>0.
\)
Then
\[
\dot x_i
=
\kappa x_i
\sum_{j\neq i}
x_j
\Bigl[
\mathbf 1_{\{\rho_i-\rho_j>\varepsilon\}}
-
\mathbf 1_{\{\rho_j-\rho_i>\varepsilon\}}
\Bigr].
\]

Here the vector field is piecewise constant in payoff space.
Population mass moves only when payoff advantages exceed the
threshold $\varepsilon$, creating regions of local inertia
in which small payoff differences generate no adjustment.
This discontinuous rule is an illustrative adjustment protocol; it falls
outside Assumption~\ref{ass:mf_regularity} unless it is smoothed or
interpreted as a differential inclusion.

Let
\[
\Phi(z)
=
\kappa \frac{e^{\beta z}-1}{\beta}\mathbf 1_{\{z>0\}},
\qquad \beta>0.
\]
The induced dynamics are
\[
\dot x_i
=
\kappa x_i
\sum_{j\neq i}
x_j
\left[
\frac{e^{\beta(\rho_i-\rho_j)}-1}{\beta}\mathbf 1_{\{\rho_i>\rho_j\}}
-
\frac{e^{\beta(\rho_j-\rho_i)}-1}{\beta}\mathbf 1_{\{\rho_j>\rho_i\}}
\right].
\]

Using the expansion
\[
\frac{e^{\beta z}-1}{\beta}=z+O(\beta),
\]
we see that as $\beta\to 0$ the dynamics converge to the replicator
equation. For large $\beta$, switching becomes highly sensitive to
large payoff gaps, generating steep directional flows.

All the above dynamics share two fundamental properties:
(i) the growth rate of $x_i$ is proportional to its current mass,
and (ii) total mass is conserved along trajectories.
Consequently, they admit a common multiplicative representation
on the simplex. This motivates the following abstract formulation.
\begin{definition}
\label{def:G_dynamics}
Let $\Delta_m=\{x\in\mathbb R_+^m:\sum_{i=1}^m x_i=1\}$.
A population dynamic on $\Delta_m$ is called a
\emph{$G$-dynamic} \cite{vincent2005evolutionary,
zhuTembineBasar2011evolutionaryMAC} if it admits the multiplicative form
\[
\dot x_i = x_i G_i(x),
\qquad i=1,\dots,m,
\]
where $G:\Delta_m\to\mathbb R^m$ is a measurable vector field satisfying
the balance condition
\[
\sum_{i=1}^m x_i G_i(x)=0
\qquad \text{for all } x\in\Delta_m.
\]
\end{definition}

\begin{proposition}
\label{prop:G_structure}
Every $G$-dynamic leaves the simplex $\Delta_m$ forward invariant.
Moreover, if $x^\ast\in\Delta_m$ is a rest point, then
\[
x_i^\ast>0
\quad\Longrightarrow\quad
G_i(x^\ast)=0.
\]
In the replicator case
\[
G_i(x)=\kappa\bigl(\rho_i(x)-\bar\rho(x)\bigr),
\]
rest points satisfy payoff equalization across all coalitions with
positive mass:
\[
x_i^\ast>0,\;x_j^\ast>0
\quad\Longrightarrow\quad
\rho_i(x^\ast)=\rho_j(x^\ast).
\]
\end{proposition}

\begin{proof}
Summing the differential equation yields
\[
\frac{d}{dt}\sum_{i=1}^m x_i(t)
=
\sum_{i=1}^m x_i(t)G_i(x(t))
=
0,
\]
so $\sum_i x_i(t)$ is constant and equals one if it does initially.
Nonnegativity follows from the multiplicative structure,
since $x_i=0$ implies $\dot x_i=0$.

If $x^\ast$ is a rest point, then $\dot x_i=0$ for all $i$.
Hence $x_i^\ast G_i(x^\ast)=0$ for all $i$,
which implies $G_i(x^\ast)=0$ whenever $x_i^\ast>0$.
For the replicator field, $G_i(x^\ast)=0$ on the active support implies
$\rho_i(x^\ast)=\bar\rho(x^\ast)$ for every active $i$, which gives
payoff equalization across active coalitions.
\end{proof}

\begin{remark}
Replicator dynamics correspond to the special case
\[
G_i(x)=\kappa\bigl(\rho_i(x)-\bar\rho(x)\bigr),
\]
but many other incentive-compatible switching rules
generate $G$-dynamics.
Thus replicator dynamics are one representative of a broader
class of mean-field exit-and-join population systems.
\end{remark}

\section{Exit-and-Join Equilibrium 
and Stationarity}

\subsection{Equilibrium and Stationarity}

We now formalize the equilibrium concept induced by the
exit-and-join dynamics and establish its equivalence
with stationarity of the mean-field system.

\begin{definition}
\label{def:continuum_equilibrium}
A population state $x^\star\in\Delta_m$ is an
\emph{exit-and-join equilibrium} if no coalition
with positive mass admits a strictly profitable deviation.
Equivalently,
\[
\rho_j(x^\star)\le \rho_i(x^\star)
\qquad
\text{for all } i \text{ with } x_i^\star>0
\text{ and all } j.
\]
\end{definition}

Thus, at equilibrium, every populated coalition
offers the same maximal payoff level,
and no coalition, populated or empty, offers a strictly higher payoff.

\begin{theorem}
\label{thm:equilibrium_characterization}
A state $x^\star\in\Delta_m$ is an exit-and-join equilibrium
if and only if there exists $\rho^\star\in\mathbb R$ such that
\[
\rho_i(x^\star)=\rho^\star
\quad \text{for all } i \text{ with } x_i^\star>0,
\]
and
\[
\rho_j(x^\star)\le \rho^\star
\quad \text{for all } j \text{ with } x_j^\star=0.
\]
\end{theorem}

\begin{proof}
Suppose $x^\star$ is an equilibrium.
For any $i,k$ with $x_i^\star>0$ and $x_k^\star>0$,
we have both
$\rho_k(x^\star)\le\rho_i(x^\star)$
and
$\rho_i(x^\star)\le\rho_k(x^\star)$,
hence equality.
Denote the common value by $\rho^\star$.

If some $j$ with $x_j^\star=0$ satisfied
$\rho_j(x^\star)>\rho^\star$,
agents in any populated coalition
would have a profitable deviation,
contradicting equilibrium.
Thus $\rho_j(x^\star)\le\rho^\star$.

Conversely, if these conditions hold,
no agent in a populated coalition
can strictly improve by moving,
hence $x^\star$ is an equilibrium.
\end{proof}

The equilibrium concept admits an exact dynamical characterization.

\begin{theorem}
\label{thm:equilibrium_stationary}
Assume the switching intensities are incentive compatible and strictly
payoff-responsive in the sense of
Assumption~\ref{ass:incentive_compatibility}.
For $x^\star\in\Delta_m$, the following are equivalent:
\begin{enumerate}
\item[(1)] $x^\star$ is an exit-and-join equilibrium;
\item[(2)] no admissible positive–measure deviation exists at $x^\star$;
\item[(3)] $x^\star$ is a stationary point of the exit-and-join dynamics,
\[
\dot x(t)=0 \quad \text{whenever } x(t)=x^\star;
\]
\item[(4)] all active incentive-compatible transition fluxes vanish at $x^\star$,
\[
x_i^\star\lambda_{ij}(x^\star)=0 \quad \text{for all } i,j.
\]
\end{enumerate}
\end{theorem}

\begin{proof}
Suppose first that $x^\star$ is an equilibrium.
Then for every $i$ with $x_i^\star>0$
and every $j$,
$\rho_j(x^\star)\le\rho_i(x^\star)$.
By incentive compatibility,
$\lambda_{ij}(x^\star)=0$.
If $x_i^\star=0$, then
$x_i^\star \lambda_{ij}(x^\star)=0$ anyway.

Next, if all active transition fluxes vanish,
all inflow and outflow terms in the mass dynamics vanish,
hence $\dot x(t)=0$.

To prove the converse implication from stationarity, argue by
contrapositive. Suppose $x^\star$ is not an
exit-and-join equilibrium. Then there exist an active coalition
$i$ with $x_i^\star>0$ and a coalition $j$ such that
$\rho_j(x^\star)>\rho_i(x^\star)$. Let $k$ be an active coalition
with minimal payoff among active coalitions:
\[
\rho_k(x^\star)
=
\min_{\ell:\,x_\ell^\star>0}\rho_\ell(x^\star).
\]
Then $\rho_j(x^\star)>\rho_k(x^\star)$. By strict
payoff-responsiveness, $\lambda_{kj}(x^\star)>0$, so there is
positive active outflow from coalition $k$ to coalition $j$.
No active coalition has payoff strictly below $\rho_k(x^\star)$,
and incentive compatibility rules out inflows to $k$ from coalitions
with weakly higher payoff. Therefore every inflow term into $k$ is
zero, while at least one outflow term is strictly positive. Hence
\[
\dot x_k
=
\sum_{\ell\neq k}\lambda_{\ell k}(x^\star)x_\ell^\star
-
\sum_{\ell\neq k}\lambda_{k\ell}(x^\star)x_k^\star
<0,
\]
so $x^\star$ is not stationary. Thus stationarity implies that no
admissible positive-measure deviation exists.

The implication from the absence of admissible positive-measure deviations
to equilibrium is immediate from the definition
of equilibrium.
\end{proof}

\begin{remark}
Strict payoff-responsiveness is essential for the converse direction.
For pairwise imitation specifications such as
$\lambda_{ij}(x)=\kappa x_j[\rho_j(x)-\rho_i(x)]_+$, an empty
destination coalition has $x_j=0$ and therefore cannot be entered.
Such dynamics may have stationary boundary states with an unused
higher-payoff coalition. On the relative interior of a fixed support,
or after adding an exploration term that permits entry into empty
coalitions, the strict-responsiveness condition is restored.
\end{remark}

\begin{remark}
An exit-and-join equilibrium equalizes payoff levels
across all coalitions with positive mass.
Coalition sizes need not be equal,
and empty coalitions may coexist with populated ones.
\end{remark}

\subsection{Mass–Based Cooperative Games and Lyapunov Structure}
\label{subsec:mass_based_lyapunov}

We now specialize the exit-and-join framework to a class of
cooperative games in which coalition value depends only on
coalition mass. In this setting, the induced population dynamics
admit a natural global Lyapunov function derived directly from
the primitive game.

\begin{assumption}
\label{ass:mass_based_game}
There exists a function $F\in C^1([0,1])$ with $F(0)=0$
such that the cooperative game satisfies
\[
v(S)=F(\mu(S)),
\qquad S\in\mathcal I.
\]
\end{assumption}

Under this assumption, coalition value depends only on its measure
and not on its composition.

Let $T=\{T_1,\dots,T_m\}$ be a coalition structure and
define the population state $x\in\Delta_m$ by
$x_i=\mu(T_i)$.
Since the restricted game $v_{|T_i}$ is again mass–based,
its total value equals $F(x_i)$.
Applying the continuum Aumann–Shapley formula to $v_{|T_i}$
yields the payoff density assigned to almost every agent in $T_i$:
\[
\rho_i(x)
=
\frac{1}{x_i}\int_0^{x_i} F'(s)\,ds,
\qquad x_i>0.
\]

By continuity of $F'$, the limit
\[
\rho_i(0):=\lim_{x_i\downarrow 0}\rho_i(x_i)=F'(0)
\]
exists, so $\rho_i$ extends continuously to $[0,1]$.

Thus, under the mass–based structure,
the payoff vector $\rho(x)$ depends only on coalition masses.

Define $V:\Delta_m\to\mathbb R$ by
\begin{equation}
\label{eq:V_definition}
V(x)
=
\sum_{i=1}^m
\int_0^{x_i}
\rho_i(s)\,ds.
\end{equation}

Since each $\rho_i(\cdot)$ is continuous on $[0,1]$,
$V$ is continuously differentiable on $\Delta_m$.
Moreover, by the fundamental theorem of calculus,
\[
\frac{\partial V(x)}{\partial x_i}
=
\rho_i(x),
\qquad i=1,\dots,m.
\]

Hence
\[
\nabla V(x)=\rho(x).
\]

\medskip

\noindent
 The function $V$ aggregates the marginal payoff levels across
coalitions and therefore represents total cooperative surplus
consistent with the Aumann–Dr\`eze allocation.
The exit-and-join dynamics will be shown to ascend $V$,
aligning individual payoff improvements with global surplus
maximization.

\subsection{Lyapunov Structure and Global Convergence}
\label{subsec:lyapunov_global}

We now establish a complete Lyapunov analysis of the exit-and-join
dynamics in the class of mass-based cooperative games and derive
global convergence results.

Throughout this subsection, Assumption~\ref{ass:mass_based_game}
remains in force. The continuum Aumann-Dr\`eze payoff level of
a coalition of mass $x_i>0$ is
\[
\rho_i(x)
=
\frac{1}{x_i}\int_0^{x_i} F'(s)\,ds,
\]
and extends continuously to $x_i=0$ with
$\rho_i(0)=F'(0)$.

Define
\begin{equation}
\label{eq:lyapunov_def}
V(x)
=
\sum_{i=1}^m
\int_0^{x_i} \rho_i(s)\,ds,
\qquad x\in\Delta_m.
\end{equation}

\begin{lemma}
\label{lem:V_regular}
The function $V$ belongs to $C^1(\Delta_m)$ and satisfies
\[
\nabla V(x)=\rho(x).
\]
\end{lemma}

\begin{proof}
Since $F\in C^1([0,1])$, the function
\[
\rho_i(x_i)=\frac{1}{x_i}\int_0^{x_i}F'(s)\,ds
\]
is continuous on $[0,1]$.
By the fundamental theorem of calculus,
\[
\frac{\partial V(x)}{\partial x_i}
=
\rho_i(x_i),
\]
which establishes differentiability and yields
$\nabla V(x)=\rho(x)$.
\end{proof}

\begin{theorem}
\label{thm:lyapunov_monotonicity}
Assume Assumption~\ref{ass:mass_based_game} and suppose the switching
intensities are incentive compatible. Then the function $V$ defined in
\eqref{eq:lyapunov_def}
is a Lyapunov function for the exit-and-join dynamics.
For every absolutely continuous solution $x(t)$,
\[
\frac{d}{dt}V(x(t))\ge 0.
\]
If the switching intensities are also strictly payoff-responsive, then
equality holds at a state $x(t)$ if and only if $x(t)$ is an
exit-and-join equilibrium.
\end{theorem}

\begin{proof}
Let $x(t)$ satisfy the exit-and-join dynamics.
By the chain rule and Lemma~\ref{lem:V_regular},
\[
\frac{d}{dt}V(x(t))
=
\nabla V(x(t))^\top \dot x(t)
=
\sum_{i=1}^m \rho_i(x(t))\,\dot x_i(t).
\]

Substituting the mass dynamics
\[
\dot{x}_i
=
\sum_{j\neq i}\lambda_{ji}(x)\,x_j
-
\sum_{j\neq i}\lambda_{ij}(x)\,x_i
\]
and rearranging terms yields
\[
\frac{d}{dt}V(x)
=
\sum_{i\neq j}
\lambda_{ij}(x)\,x_i\,
\bigl[\rho_j(x)-\rho_i(x)\bigr].
\]

By incentive compatibility,
\[
\lambda_{ij}(x)>0
\;\Rightarrow\;
\rho_j(x)>\rho_i(x).
\]
Since $x_i\ge 0$,
each term in the sum is nonnegative,
hence $dV(x(t))/dt\ge 0$.

Moreover, under strict payoff-responsiveness,
\[
\frac{d}{dt}V(x(t))=0
\]
if and only if
\[
\lambda_{ij}(x(t))\,x_i(t)=0
\quad \text{for all } i,j.
\]
This condition is equivalent to
\[
\rho_j(x(t))\le \rho_i(x(t))
\quad
\text{for all } i \text{ with } x_i(t)>0,
\]
which is precisely the exit-and-join equilibrium condition.
\end{proof}

\subsection{Global Convergence}

\begin{assumption}
\label{ass:strict_concavity}
The function $F$ is strictly concave on $[0,1]$.
\end{assumption}

Under strict concavity of $F$, coalition marginal productivity
is strictly decreasing in coalition mass.

\begin{theorem}
\label{thm:global_convergence_complete}
Suppose Assumptions~\ref{ass:mass_based_game}
and~\ref{ass:strict_concavity} hold. Suppose also that the switching
intensities satisfy Assumption~\ref{ass:mf_regularity} and are incentive
compatible and strictly payoff-responsive.
Then the exit-and-join dynamics admit a unique equilibrium
$x^\star\in\Delta_m$, and for every initial condition
$x(0)\in\Delta_m$,
\[
\lim_{t\to\infty} x(t)=x^\star.
\]
\end{theorem}

\begin{proof}
Strict concavity of $F$ implies that $F'$ is strictly decreasing on $[0,1]$.
Since
\[
\rho_i(x_i)
=
\frac{1}{x_i}\int_0^{x_i} F'(s)\,ds,
\]
the function $x_i \mapsto \rho_i(x_i)$ is strictly decreasing on $(0,1]$.
Consequently, the mapping
\[
x_i \longmapsto \int_0^{x_i} \rho_i(s)\,ds
\]
is strictly concave on $[0,1]$.
Because $V$ is the sum of these coordinate functions and
$\Delta_m$ is convex, it follows that $V$ is strictly concave on $\Delta_m$.

Since $\Delta_m$ is compact and convex and $V$ is continuous,
$V$ attains a maximizer $x^\star \in \Delta_m$.
Strict concavity guarantees that this maximizer is unique.
The first-order optimality condition for maximizing the differentiable
concave function $V$ over $\Delta_m$ is
\[
\rho_j(x^\star)\le \rho_i(x^\star)
\qquad
\text{for all } i \text{ with } x_i^\star>0
\text{ and all } j,
\]
which is exactly the exit-and-join equilibrium condition. Hence
$x^\star$ is the unique equilibrium.

For any absolutely continuous solution $x(t)$ of the dynamics,
Theorem~\ref{thm:lyapunov_monotonicity} implies
\[
\frac{d}{dt}V(x(t)) \ge 0,
\]
with strict inequality whenever $x(t)$ is not an equilibrium.
Thus $V(x(t))$ is nondecreasing along trajectories and strictly
increasing outside $x^\star$.
Since $V$ is continuous on the compact set $\Delta_m$,
it is bounded above, and therefore
\[
\lim_{t\to\infty} V(x(t))
\]
exists.

The set $\{x\in\Delta_m : \dot V(x)=0\}$ coincides with the set of equilibria.
Because the equilibrium is unique, this invariant set reduces to
$\{x^\star\}$.
By LaSalle’s invariance principle, every trajectory converges to
$x^\star$, completing the proof.
\end{proof}

\begin{remark}
If $F$ is concave but not strictly concave,
then $V$ is concave but may admit multiple maximizers.
In this case, every trajectory converges to the compact
set of equilibria,
but convergence need not be to a unique point.
\end{remark}

\begin{remark}
The Lyapunov function $V$ represents aggregate cooperative surplus.
The exit-and-join dynamics implement a decentralized,
incentive-compatible ascent of total surplus.
Strict concavity ensures that surplus maximization
selects a unique coalition size distribution,
so individual rational adjustments lead globally
to the socially efficient allocation.
\end{remark}

\section{Population Game Formulation and Wardrop Equivalence}
\label{sec:wardrop_equivalence}

This section establishes that exit--and--join equilibria coincide with
Wardrop equilibria of an induced nonatomic population game.
The cooperative structure determines the payoff vector field $\rho(x)$,
while the equilibrium concept itself is entirely noncooperative. This
link places the model alongside Wardrop formulations of traffic
assignment and recent security and learning variants of congestion games
\cite{Wardrop1952,panZhu2022poisonedWardrop,panLiZhu2022resilienceTraffic}.

\subsection{The Induced Population Game}

Let $(I,\mathcal I,\mu)$ be a nonatomic probability space with $\mu(I)=1$.
Fix $m<\infty$ and define the simplex
\[
\Delta_m
=
\left\{
x \in \mathbb{R}^m_+ :
\sum_{i=1}^m x_i = 1
\right\}.
\]

From Section~\ref{sec:continuum_exit_join},
the continuum Aumann--Dr\`eze construction induces a payoff vector field
\[
\rho : \Delta_m \to \mathbb{R}^m,
\qquad
\rho(x) = (\rho_1(x),\dots,\rho_m(x)),
\]
where $\rho_i(x)$ denotes the payoff density assigned to coalition $i$
when the population state is $x$.

\begin{definition}
\label{def:population_game}
The induced nonatomic population game is the triple $G=(I,A,u)$ defined as follows.
The player set is $I$. The strategy set is $A=\{1,\dots,m\}$. 
A (pure) strategy profile is a measurable map $\sigma:I\to A$. 
The induced population state is given by 
$x_i(\sigma)=\mu(\{k\in I:\sigma(k)=i\})$ for $i=1,\dots,m$, 
so that $x(\sigma)=(x_1(\sigma),\dots,x_m(\sigma))\in\Delta_m$. 
If the population state is $x\in\Delta_m$, the payoff to a player choosing strategy $i\in A$ is $u_i(x)=\rho_i(x)$.
\end{definition}

Thus payoffs depend only on aggregate coalition masses.

\subsection{Wardrop Equilibrium}

We now introduce the appropriate equilibrium notion for nonatomic populations.

\begin{definition}
\label{def:wardrop}
A state $x^\star \in \Delta_m$ is a Wardrop equilibrium of $G$
if
\[
x_i^\star > 0
\;\Longrightarrow\;
u_i(x^\star)
=
\max_{j=1,\dots,m} u_j(x^\star).
\]
Equivalently,
\[
u_j(x^\star)
\le
u_i(x^\star)
\quad
\text{for all } i \text{ with } x_i^\star>0.
\]
\end{definition}

Since $u_i(x)=\rho_i(x)$, the Wardrop condition can be written as
\[
\rho_j(x^\star)
\le
\rho_i(x^\star)
\quad
\forall i \text{ with } x_i^\star>0.
\]

\subsection{Equivalence with Exit--and--Join Equilibrium}

Recall from Definition~\ref{def:continuum_equilibrium} that
$x^\star\in\Delta_m$ is an exit--and--join equilibrium if
\[
\rho_j(x^\star)
\le
\rho_i(x^\star)
\quad
\forall i \text{ with } x_i^\star>0.
\]

\begin{theorem}
\label{thm:wardrop_equivalence}
Let $\rho : \Delta_m \to \mathbb{R}^m$ be the payoff vector
induced by the continuum Aumann--Dr\`eze value.
Then for $x^\star\in\Delta_m$ the following are equivalent:
\begin{enumerate}
\item $x^\star$ is an exit--and--join equilibrium;
\item $x^\star$ is a Wardrop equilibrium of the population game $G$.
\end{enumerate}
\end{theorem}

\begin{proof}
By Definition~\ref{def:population_game},
the payoff to strategy $i$ is $u_i(x)=\rho_i(x)$.
Hence the Wardrop condition in Definition~\ref{def:wardrop}
is identical to the exit--and--join equilibrium condition.
\end{proof}

\subsection{Variational Inequality Characterization}

Wardrop equilibria admit an equivalent variational inequality (VI) formulation.

\begin{proposition}
\label{prop:VI}
A state $x^\star \in \Delta_m$ is a Wardrop equilibrium
if and only if
\[
\langle \rho(x^\star), x - x^\star \rangle
\le 0
\quad
\forall x \in \Delta_m.
\]
\end{proposition}

\begin{proof}
Suppose first that $x^\star$ is a Wardrop equilibrium.
Let $x\in\Delta_m$.
Partition the index set into
\[
P=\{i : x_i^\star>0\}, 
\qquad 
Z=\{i : x_i^\star=0\}.
\]
For $i\in P$, $\rho_i(x^\star)=\max_j\rho_j(x^\star)$.
For $i\in Z$, $\rho_i(x^\star)\le \max_j\rho_j(x^\star)$.
Hence
\[
\sum_{i=1}^m \rho_i(x^\star)x_i
\le
\max_j \rho_j(x^\star)\sum_{i=1}^m x_i
=
\max_j \rho_j(x^\star).
\]
Similarly,
\[
\sum_{i=1}^m \rho_i(x^\star)x_i^\star
=
\max_j \rho_j(x^\star).
\]
Subtracting yields
\[
\langle \rho(x^\star), x-x^\star\rangle \le 0.
\]

Conversely, suppose the variational inequality holds.
If there existed $i$ with $x_i^\star>0$
and $j$ with $\rho_j(x^\star)>\rho_i(x^\star)$,
consider the direction $x = x^\star + \varepsilon (e_j-e_i)$
for sufficiently small $\varepsilon>0$.
Then
\[
\langle \rho(x^\star), x-x^\star\rangle
=
\varepsilon(\rho_j(x^\star)-\rho_i(x^\star))>0,
\]
contradicting the VI condition.
Thus the Wardrop condition must hold.
\end{proof}

\subsection{Mass--Based Games and Potential Structure}

Under Assumption~\ref{ass:mass_based_game},
\[
v(S)=F(\mu(S)),
\]
the payoff field satisfies
\[
\nabla V(x)=\rho(x),
\]
where
\[
V(x)
=
\sum_{i=1}^m
\int_0^{x_i} \rho_i(s)\,ds.
\]

\begin{proposition}
In the mass--based case,
the induced population game is a potential game
with potential $V$.
Wardrop equilibria coincide with maximizers of $V$ on $\Delta_m$.
\end{proposition}

\begin{proof}
Since $\nabla V(x)=\rho(x)$,
the variational inequality in Proposition~\ref{prop:VI}
is equivalent to the first-order optimality condition
for maximizing $V$ over $\Delta_m$.
\end{proof}

\section{Switching Costs and Endogenous Acceptance Rules}
\label{sec:acceptance_rules}

In this section we extend the exit--and--join framework by incorporating
(i) switching costs and (ii) endogenous acceptance rules.
Under this extension, coalition transitions become bilateral:
a deviation from coalition $i$ to coalition $j$ must be
individually profitable and acceptable to the incumbent members
of coalition $j$.

Throughout this section, assume that the payoff field
$\rho:\Delta_m\to\mathbb R^m$ is continuously differentiable.

\subsection{Switching Costs}

For each ordered pair $(i,j)$ with $i\neq j$,
let $c_{ij}:\Delta_m\to\mathbb R_+$ denote the switching cost
incurred by an agent moving from coalition $i$ to coalition $j$.

\begin{definition}
For $x\in\Delta_m$, define
\[
\Delta_{ij}(x)
=
\rho_j(x)-\rho_i(x)-c_{ij}(x).
\]
A deviation from $i$ to $j$ satisfies
\emph{individual rationality} at $x$ if
$\Delta_{ij}(x)>0$.
\end{definition}

\subsection{Derivation of the Acceptance Rule}

Because the continuum Aumann--Dr\`eze value assigns identical payoff
density to all members of a coalition,
every incumbent member of coalition $j$
receives payoff $\rho_j(x)$.

Consider a state $x\in\Delta_m$ and an infinitesimal mass
$\varepsilon>0$ moving from coalition $i$ to coalition $j$.
The perturbed state is
\[
x^\varepsilon
=
x + \varepsilon(e_j - e_i),
\]
where $e_k$ denotes the $k$-th canonical basis vector.

\begin{definition}
Coalition $j$ accepts the entrant at state $x$
if and only if
\[
\rho_j(x^\varepsilon)
\ge
\rho_j(x)
\quad
\text{for all sufficiently small } \varepsilon>0.
\]
\end{definition}

Since $\rho$ is continuously differentiable,
the directional derivative of $\rho_j$ in direction
$d=e_j-e_i$ is
\[
D\rho_j(x)[d]
=
\nabla\rho_j(x)\cdot(e_j-e_i)
=
\frac{\partial \rho_j}{\partial x_j}(x)
-
\frac{\partial \rho_j}{\partial x_i}(x).
\]
To first order,
\[
\rho_j(x^\varepsilon)
=
\rho_j(x)
+
\varepsilon
D\rho_j(x)[e_j-e_i]
+
o(\varepsilon).
\]

Hence coalition $j$ weakly benefits from entry
if and only if
\[
D\rho_j(x)[e_j-e_i]
\ge 0.
\]

\begin{proposition}
\label{prop:acceptance_rule}
Under continuous differentiability of $\rho$,
entry from coalition $i$ into coalition $j$ at state $x$
is admissible if and only if
\[
D\rho_j(x)[e_j-e_i]\ge 0.
\]
\end{proposition}

Thus acceptance depends on the first-order effect of the proposed
transfer on the payoff of incumbent members of the destination coalition.
In the important own-mass case, where $\rho_j(x)$ depends only on $x_j$,
this condition reduces to
\[
\frac{\partial \rho_j}{\partial x_j}(x)\ge 0.
\]

\subsection{Admissible Deviations and Switching Dynamics}

For compactness, define the acceptance margin
\[
A_{ij}(x):=D\rho_j(x)[e_j-e_i].
\]

\begin{definition}
A deviation from coalition $i$ to coalition $j$
is admissible at state $x\in\Delta_m$
if $x_i>0$ and both
\[
\Delta_{ij}(x)>0
\quad\text{and}\quad
A_{ij}(x)\ge 0.
\]
\end{definition}

Thus switching requires both individual profitability
and coalitional consent.

Let $\Phi:\mathbb R\to\mathbb R_+$ satisfy
$\Phi(z)=0$ for $z\le 0$ and $\Phi(z)>0$ for $z>0$.
Define the admissible switching intensity
\[
\lambda_{ij}(x)
=
\Phi\bigl(\Delta_{ij}(x)\bigr)
\mathbf 1_{\{A_{ij}(x)\ge 0\}}.
\]

The induced mass dynamics are
\[
\dot x_i
=
\sum_{j\neq i} \lambda_{ji}(x)x_j
-
\sum_{j\neq i} \lambda_{ij}(x)x_i,
\qquad i=1,\dots,m.
\]

\begin{proposition}
These dynamics constitute a $G$-dynamic
in the sense of Definition~\ref{def:G_dynamics}.
\end{proposition}

\begin{proof}
Summing over $i$ yields
\[
\sum_{i=1}^m \dot x_i
=
\sum_{i\neq j} \lambda_{ji}(x)x_j
-
\sum_{i\neq j} \lambda_{ij}(x)x_i
=
0,
\]
so total mass is preserved.
The multiplicative structure follows by factoring $x_i$.
\end{proof}

\subsection{Constrained Equilibrium under Switching Costs and Acceptance}
\label{subsec:constrained_equilibrium}

The introduction of switching costs and endogenous acceptance rules
alters the equilibrium concept in a fundamental way.
Without these constraints, equilibrium coincides with a Wardrop
equilibrium and admits a classical variational inequality
characterization over the simplex $\Delta_m$.
With bilateral admissibility, equilibrium becomes constrained
by state-dependent feasibility conditions,
and the set of admissible deviations depends endogenously on the state.

Recall that the net gain from switching from coalition $i$ to coalition $j$ is
\[
\Delta_{ij}(x)
=
\rho_j(x)-\rho_i(x)-c_{ij}(x).
\]

A deviation from $i$ to $j$ is admissible at $x\in\Delta_m$ if $x_i>0$,
\[
\Delta_{ij}(x)>0
\quad\text{and}\quad
A_{ij}(x)\ge 0.
\]

The first condition enforces individual rationality.
The second condition ensures that coalition $j$ weakly benefits
from the proposed transfer at the margin.

\begin{definition}
For $x\in\Delta_m$, define
\[
\mathcal D(x)
=
\{\, e_j-e_i :
x_i>0,\ 
\Delta_{ij}(x)>0
\text{ and }
A_{ij}(x)\ge 0
\,\}.
\]
\end{definition}

The admissible deviation cone is
\[
\mathcal T(x)
=
\mathrm{cone}\bigl(\mathcal D(x)\bigr).
\]

Unlike the unconstrained case,
$\mathcal T(x)$ is generally a strict subset
of the tangent cone of $\Delta_m$.

\begin{definition}
A state $x^\star\in\Delta_m$ is a constrained equilibrium
if
\[
\langle \rho(x^\star), d\rangle \le 0
\quad
\text{for all } d\in\mathcal D(x^\star).
\]
Equivalently,
\[
\langle \rho(x^\star), x-x^\star\rangle \le 0
\quad
\forall x-x^\star\in\mathcal T(x^\star).
\]
\end{definition}

Thus no admissible bilateral deviation yields
a first-order payoff improvement.

\begin{proposition}
\label{prop:constrained_characterization}
A state $x^\star\in\Delta_m$ is a constrained equilibrium
if and only if for every $i$ with $x_i^\star>0$ and every $j\neq i$,
either
\[
\rho_j(x^\star)-c_{ij}(x^\star)
\le
\rho_i(x^\star),
\]
or
\[
A_{ij}(x^\star)<0.
\]
\end{proposition}

\begin{proof}
If both inequalities fail, then
$\Delta_{ij}(x^\star)>0$
and
$A_{ij}(x^\star)\ge 0$,
so $d=e_j-e_i\in\mathcal D(x^\star)$
and
\[
\langle \rho(x^\star), d\rangle
=
\rho_j(x^\star)-\rho_i(x^\star)
>0,
\]
contradicting constrained equilibrium.
The converse follows directly from the definition of $\mathcal D(x^\star)$.
\end{proof}

 \subsubsection{Relation to Wardrop Equilibrium}

Let $T_{\Delta_m}(x)$ denote the tangent cone of the simplex
$\Delta_m$ at $x$, that is,
\[
T_{\Delta_m}(x)
=
\left\{
d\in\mathbb R^m :
\sum_{i=1}^m d_i=0,
\;
d_i\ge 0 \text{ whenever } x_i=0
\right\}.
\]

\begin{proposition}
\label{prop:wardrop_reduction}
Suppose that
\begin{enumerate}
\item $c_{ij}(x)\equiv 0$ for all $i,j$,
\item $A_{ij}(x)\ge 0$ for all $i\neq j$ and all $x\in\Delta_m$.
\end{enumerate}
Then constrained equilibrium coincides with Wardrop equilibrium.
\end{proposition}

\begin{proof}
If $c_{ij}\equiv 0$ and acceptance is automatic,
then a deviation $i\to j$ is admissible whenever
$\rho_j(x)>\rho_i(x)$.
Hence
\[
\mathcal D(x)
=
\{\, e_j-e_i : \rho_j(x)>\rho_i(x)\,\}.
\]
Thus a constrained equilibrium has no elementary transfer
from a populated coalition to a coalition with strictly higher payoff.
Equivalently,
\[
\rho_j(x^\star)
\le
\rho_i(x^\star)
\quad
\forall i \text{ with } x_i^\star>0,
\]
which is precisely the Wardrop equilibrium condition.
\end{proof}

Thus Wardrop equilibrium is a special case of constrained
equilibrium under full feasibility of deviations.

\subsubsection{Modified Variational Inequality}

In the absence of switching costs and acceptance constraints,
equilibrium is characterized by the variational inequality
\[
\langle \rho(x^\star), x-x^\star\rangle \le 0
\quad
\forall x\in\Delta_m,
\]
or equivalently,
\[
\langle \rho(x^\star), d\rangle \le 0
\quad
\forall d\in T_{\Delta_m}(x^\star),
\]
where $T_{\Delta_m}(x^\star)$ denotes the tangent cone of the simplex.
This is the classical Wardrop condition.

Under switching costs and acceptance rules,
the feasible comparison set is restricted to the admissible deviation cone
$\mathcal T(x^\star)\subseteq T_{\Delta_m}(x^\star)$.
Accordingly, equilibrium satisfies the state-dependent variational inequality
\[
\langle \rho(x^\star), d\rangle \le 0
\quad
\forall d\in \mathcal T(x^\star).
\]
Since $\mathcal T(x^\star)\subseteq T_{\Delta_m}(x^\star)$ in general,
the constrained equilibrium condition is weaker than the Wardrop condition
and depends endogenously on $x^\star$.
The problem is therefore a state-dependent variational inequality,
i.e., a quasi-variational inequality (QVI).

The classical variational inequality is equivalent to the normal cone inclusion
\[
-\,\rho(x^\star)\in N_{\Delta_m}(x^\star),
\]
where $N_{\Delta_m}(x^\star)$ is the normal cone of
\[
\Delta_m=\{x\in\mathbb R_+^m:\sum_{i=1}^m x_i=1\}.
\]
Explicitly,
\[
N_{\Delta_m}(x^\star)
=
\left\{
\lambda \mathbf{1}-\mu :
\lambda\in\mathbb R,\;
\mu_i\ge 0,\;
\mu_i x_i^\star=0
\right\}.
\]
Hence Wardrop equilibrium is equivalent to the existence of
$\lambda\in\mathbb R$ such that
\[
\rho_i(x^\star)=\lambda
\quad \text{if } x_i^\star>0,
\qquad
\rho_i(x^\star)\le \lambda
\quad \text{if } x_i^\star=0.
\]
The scalar $\lambda$ is the Lagrange multiplier associated with
the mass constraint $\sum_i x_i=1$.

Under switching costs and acceptance rules,
the dual inclusion becomes
\(
-\,\rho(x^\star)\in N_{\mathcal A(x^\star)}(x^\star),
\)
where $\mathcal A(x^\star)$ denotes the locally admissible feasible set
generated by $\mathcal T(x^\star)$.
Because $\mathcal A(x^\star)$ depends on the state,
the associated normal cone is state-dependent,
and the equilibrium condition is a quasi-variational inequality.

Since $\mathcal T(x^\star)\subseteq T_{\Delta_m}(x^\star)$,
we have
\(
N_{\Delta_m}(x^\star)
\subseteq
N_{\mathcal A(x^\star)}(x^\star),
\)
so constrained equilibrium allows payoff differentials
that are ruled out under Wardrop equilibrium.

\subsubsection{Structural differences from Wardrop equilibrium.}

The constrained equilibrium differs from the classical Wardrop equilibrium
in three essential respects.
Switching costs introduce wedges between payoffs at equilibrium.
In particular, it may occur that
$\rho_j(x^\star)>\rho_i(x^\star)$
for some $i$ with $x_i^\star>0$,
provided that
\[
\rho_j(x^\star)-\rho_i(x^\star)
\le
c_{ij}(x^\star).
\]
Thus payoff equalization across active coalitions
is no longer necessary.
Equilibrium permits bounded payoff differentials
that are sustained by switching frictions.
Acceptance rules generate endogenous entry restrictions.
If
\[
A_{ij}(x^\star)<0,
\]
coalition $j$ rejects entry from coalition $i$ at the margin.
Consequently, coalitions may stabilize at interior sizes,
even when $\rho_j(x^\star)$ exceeds the payoff of other coalitions.
Equilibrium coalition sizes are therefore determined not only by payoff levels,
but also by the marginal effect of transfers on coalition payoff.
The admissible deviation cone $\mathcal T(x^\star)$
depends explicitly on the equilibrium state.
As a result, equilibrium is characterized by a
state-dependent constrained variational inequality.
Unlike the Wardrop case—where the feasible deviation set
is the tangent cone of $\Delta_m$—
the feasible directions here are determined endogenously
by switching costs and marginal acceptance conditions.

\subsubsection{Special Case: Concave Mass-Based Games}

Consider the mass-based case in which
\(
v(S)=F(\mu(S)),
\)
and assume $F:\mathbb R_+\to\mathbb R$ is twice continuously differentiable.
Recall that the induced payoff density satisfies
\[
\rho_j(x_j)
=
\frac{1}{x_j}
\int_0^{x_j} F'(s)\,ds,
\qquad x_j>0.
\]

\begin{proposition}
\label{prop:concave_rejection}
If $F$ is strictly concave, i.e., $F''<0$,
then for every interior state $x\in\Delta_m$
and every destination coalition $j$ with $x_j>0$,
\[
A_{ij}(x)<0
\quad\text{for every } i\neq j.
\]
\end{proposition}

\begin{proof}
Define
\(
R(x_j)
=
\frac{1}{x_j}
\int_0^{x_j} F'(s)\,ds.
\)
Differentiating with respect to $x_j$,
\[
\frac{dR}{dx_j}
=
\frac{x_j F'(x_j)-\int_0^{x_j}F'(s)\,ds}{x_j^2}.
\]
By strict concavity of $F$,
$F'$ is strictly decreasing,
so
\(
\int_0^{x_j}F'(s)\,ds
>
x_j F'(x_j).
\)
Hence the numerator is negative,
and therefore
\(
\frac{\partial \rho_j}{\partial x_j}(x)
=
\frac{dR}{dx_j}
<0.
\)
Since the mass-based payoff $\rho_j$ depends only on $x_j$,
we have $A_{ij}(x)=D\rho_j(x)[e_j-e_i]=dR/dx_j<0$.
\end{proof}

By Proposition~\ref{prop:concave_rejection},
interior coalitions strictly reject entry.
Thus for any interior state,
\[
A_{ij}(x)<0
\quad\text{whenever } x_j>0.
\]
Consequently, no admissible expansion of an interior coalition is possible.

\begin{theorem}
\label{thm:concave_acceptance_mechanism}
Consider the mass-based case $v(S)=F(\mu(S))$.
Assume $F$ is strictly concave and $c_{ij}\equiv 0$.
Adopt the convention that an empty coalition has no incumbents
and therefore accepts first entry automatically.

Then:

\begin{enumerate}
\item For every interior state $x\in\Delta_m$ and every $j$ with $x_j>0$,
\[
A_{ij}(x)<0
\quad\text{for all } i\neq j.
\]
Hence every active coalition strictly rejects marginal entry.

\item If $x$ is an interior state, then no admissible deviation exists.
Thus every interior state is a constrained equilibrium under the
acceptance rule, even though it need not be a Wardrop equilibrium.

\item At a boundary state $x^\star$, admissible deviations can only
target empty coalitions. Consequently, $x^\star$ is a constrained
equilibrium if and only if
\[
\rho_j(x^\star)\le \rho_i(x^\star)
\quad
\text{for every } i \text{ with } x_i^\star>0
\text{ and every } j \text{ with } x_j^\star=0.
\]

\item Therefore acceptance-constrained equilibrium is generally weaker
than Wardrop equilibrium and need not coincide with the maximizers of
the potential $V$ unless additional acceptance or feasibility assumptions
restore all payoff-improving directions.
\end{enumerate}
\end{theorem}

\begin{proof}

In the mass-based case $v(S)=F(\mu(S))$, the induced payoff density is
\[
\rho_j(x_j)
=
\frac{1}{x_j}\int_0^{x_j} F'(s)\,ds,
\qquad x_j>0.
\]
Define the potential
\[
V(x)
=
\sum_{i=1}^m \int_0^{x_i}\rho_i(s)\,ds.
\]
A direct differentiation shows that $\nabla V(x)=\rho(x)$.

Assume $F$ is strictly concave, so $F''<0$ and $F'$ is strictly decreasing.
Differentiating $\rho_j$ yields
\[
\frac{\partial \rho_j}{\partial x_j}(x)
=
\frac{x_j F'(x_j)-\int_0^{x_j}F'(s)\,ds}{x_j^2}.
\]
Since $F'$ is strictly decreasing,
\[
\int_0^{x_j}F'(s)\,ds
>
x_j F'(x_j),
\]
which implies
\[
\frac{\partial \rho_j}{\partial x_j}(x)<0
\quad
\text{whenever } x_j>0.
\]
Because $\rho_j$ depends only on $x_j$, this is exactly
$A_{ij}(x)<0$ for every $i\neq j$, proving statement (1).

If $x$ is interior, every destination coalition has positive mass.
By statement (1), every possible destination rejects marginal entry.
Hence $\mathcal D(x)=\varnothing$, proving statement (2).

At a boundary state, any admissible deviation $i\to j$ must have
$x_i^\star>0$ and, by statement (1), cannot have $x_j^\star>0$.
Thus it can only target an empty coalition. Since switching costs are
zero and empty coalitions accept first entry by convention, such a
deviation is admissible exactly when
$\rho_j(x^\star)>\rho_i(x^\star)$. Absence of admissible deviations is
therefore equivalent to the inequality in statement (3).

Wardrop equilibrium requires the same inequality against all coalitions,
not only empty destinations. The acceptance rule removes all transfers
into active coalitions from the feasible deviation set, so the constrained
condition is generally weaker than Wardrop equilibrium and weaker than the
first-order optimality condition for maximizing $V$ over the whole simplex.
This proves statement (4).

\end{proof}

\section{Numerical Studies in the Large-Population Regime}
\label{sec:numerical_studies}

This section illustrates the exit-and-join mechanism in a large
population. The examples are not intended as calibration exercises.
They are designed to make visible the qualitative mechanisms established
above: payoff-responsive mass movement, convergence of the finite
population process to the deterministic mean-field dynamics, and the
persistence of payoff gaps when switching frictions constrain mobility.

We consider four coalitions with payoff densities
\[
\rho_i(x)=\theta_i-\beta_i x_i,\qquad i=1,\dots,4,
\]
where $\theta=(1.25,1.10,0.92,0.80)$ and
$\beta=(1.80,1.20,0.90,0.70)$. This specification captures diminishing
returns within each coalition and is the gradient of the concave
potential $V(x)=\sum_i \theta_i x_i-\frac12\sum_i\beta_i x_i^2$.
The initial population state is
$x(0)=(0.62,0.24,0.10,0.04)$, so the first coalition begins large but
has low payoff because of congestion. Unless otherwise stated, switching
uses the payoff-difference intensity
\[
\lambda_{ij}(x)=\kappa x_j[\rho_j(x)-\rho_i(x)]_+,
\qquad \kappa=3.
\]
The resulting mean-field dynamics are the replicator-type
exit-and-join dynamics derived in Section~\ref{sec:continuum_exit_join}.
For the chosen parameters, the unique interior equilibrium and common
active payoff are
\[
x^\star\approx(0.302,0.328,0.237,0.133),
\qquad
\rho^\star\approx0.707.
\]

\begin{figure}[t]
\centering
\begin{subfigure}[t]{0.32\textwidth}
\centering
\begin{tikzpicture}
\begin{axis}[
width=\linewidth,
height=3.6cm,
xlabel={$t$},
ylabel={$x_i(t)$},
xmin=0,xmax=16,
ymin=0,ymax=0.7,
tick label style={font=\scriptsize},
label style={font=\scriptsize},
legend style={font=\scriptsize,draw=none,fill=none,at={(0.98,0.98)},anchor=north east},
grid=both,
grid style={line width=.1pt,draw=gray!20}
]
\addplot[blue,thick] table[x=t,y=x1,col sep=comma] {figures/mean_field.csv};
\addlegendentry{$C_1$}
\addplot[red,thick] table[x=t,y=x2,col sep=comma] {figures/mean_field.csv};
\addlegendentry{$C_2$}
\addplot[green!60!black,thick] table[x=t,y=x3,col sep=comma] {figures/mean_field.csv};
\addlegendentry{$C_3$}
\addplot[purple,thick] table[x=t,y=x4,col sep=comma] {figures/mean_field.csv};
\addlegendentry{$C_4$}
\end{axis}
\end{tikzpicture}
\caption{Coalition masses.}
\end{subfigure}
\hfill
\begin{subfigure}[t]{0.32\textwidth}
\centering
\begin{tikzpicture}
\begin{axis}[
width=\linewidth,
height=3.6cm,
xlabel={$t$},
ylabel={$\rho_i(x(t))$},
xmin=0,xmax=16,
ymin=0.1,ymax=0.9,
tick label style={font=\scriptsize},
label style={font=\scriptsize},
legend style={font=\scriptsize,draw=none,fill=none,at={(0.98,0.98)},anchor=north east},
grid=both,
grid style={line width=.1pt,draw=gray!20}
]
\addplot[blue,thick] table[x=t,y=rho1,col sep=comma] {figures/mean_field.csv};
\addplot[red,thick] table[x=t,y=rho2,col sep=comma] {figures/mean_field.csv};
\addplot[green!60!black,thick] table[x=t,y=rho3,col sep=comma] {figures/mean_field.csv};
\addplot[purple,thick] table[x=t,y=rho4,col sep=comma] {figures/mean_field.csv};
\addplot[black,dotted,thick] coordinates {(0,0.7067) (16,0.7067)};
\addlegendentry{common payoff}
\end{axis}
\end{tikzpicture}
\caption{Payoff equalization.}
\end{subfigure}
\hfill
\begin{subfigure}[t]{0.32\textwidth}
\centering
\begin{tikzpicture}
\begin{axis}[
width=\linewidth,
height=3.6cm,
xlabel={$t$},
ylabel={$V(x(t))$},
xmin=0,xmax=16,
ymin=0.77,ymax=0.92,
tick label style={font=\scriptsize},
label style={font=\scriptsize},
grid=both,
grid style={line width=.1pt,draw=gray!20}
]
\addplot[black,thick] table[x=t,y=potential,col sep=comma] {figures/mean_field.csv};
\end{axis}
\end{tikzpicture}
\caption{Surplus ascent.}
\end{subfigure}
\caption{Mean-field exit-and-join dynamics. Mass leaves the initially
congested coalition and reallocates toward coalitions with higher payoff
density. The active payoff densities equalize at equilibrium, while the
potential increases monotonically along the trajectory.}
\label{fig:mean_field_exit_join}
\end{figure}
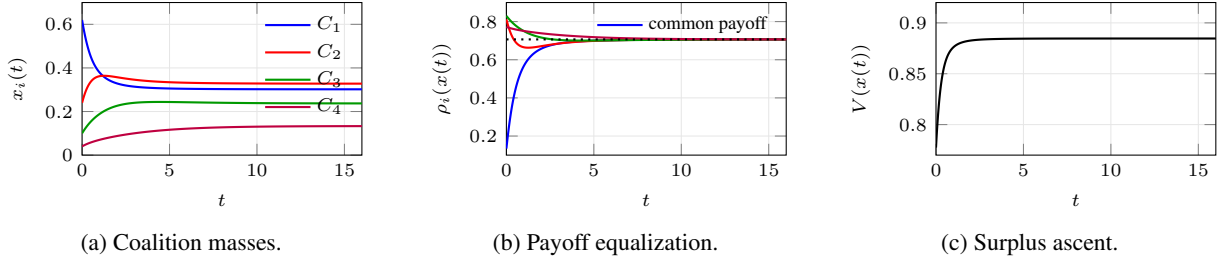

Figure~\ref{fig:mean_field_exit_join} shows the deterministic
large-population dynamics. The first coalition initially has the largest
mass but the lowest payoff density, so it loses agents. The other
coalitions attract mass until the payoff densities are equalized. The
potential rises throughout the adjustment, illustrating the Lyapunov
structure in Theorem~\ref{thm:lyapunov_monotonicity}.

To connect the deterministic equation with a finite but large
population, let $n_i^N(t)$ be the number of agents in coalition $i$ and
let $x_i^N(t)=n_i^N(t)/N$. Over a short time interval, an agent in
coalition $i$ switches to coalition $j$ with probability approximately
$\lambda_{ij}(x^N(t))\Delta t$. This finite population Markov process
has drift equal to the mean-field vector field, and its fluctuations
vanish at the usual order $N^{-1/2}$, consistent with classical
law-of-large-numbers approximations for density-dependent population
processes \cite{ethier2009markov,sandholm2010population}.

\begin{figure}[t]
\centering
\begin{subfigure}[t]{0.48\textwidth}
\centering
\begin{tikzpicture}
\begin{axis}[
width=\linewidth,
height=4.2cm,
xlabel={$t$},
ylabel={$\|x^N(t)-x(t)\|_2$},
xmin=0,xmax=8,
ymin=0,ymax=0.075,
tick label style={font=\scriptsize},
label style={font=\scriptsize},
legend style={font=\scriptsize,draw=none,fill=none,at={(0.98,0.98)},anchor=north east},
grid=both,
grid style={line width=.1pt,draw=gray!20}
]
\addplot[blue,thick] table[x=t,y=err200,col sep=comma] {figures/large_population_error.csv};
\addlegendentry{$N=200$}
\addplot[red,thick] table[x=t,y=err1000,col sep=comma] {figures/large_population_error.csv};
\addlegendentry{$N=1000$}
\addplot[green!60!black,thick] table[x=t,y=err5000,col sep=comma] {figures/large_population_error.csv};
\addlegendentry{$N=5000$}
\end{axis}
\end{tikzpicture}
\caption{Mean path error.}
\end{subfigure}
\hfill
\begin{subfigure}[t]{0.48\textwidth}
\centering
\begin{tikzpicture}
\begin{axis}[
width=\linewidth,
height=4.2cm,
xlabel={$N$},
ylabel={mean sup error},
xmode=log,
ymode=log,
tick label style={font=\scriptsize},
label style={font=\scriptsize},
legend style={font=\scriptsize,draw=none,fill=none,at={(0.98,0.98)},anchor=north east},
grid=both,
grid style={line width=.1pt,draw=gray!20}
]
\addplot[black,mark=*,thick] table[x=N,y=sup_error,col sep=comma] {figures/large_population_scaling.csv};
\addlegendentry{simulation}
\addplot[gray,dashed,thick] table[x=N,y=reference,col sep=comma] {figures/large_population_scaling.csv};
\addlegendentry{$N^{-1/2}$}
\end{axis}
\end{tikzpicture}
\caption{Error scaling.}
\end{subfigure}
\caption{Finite-population approximation. The stochastic
exit-and-join process concentrates around the deterministic trajectory
as $N$ grows, and the mean sup-norm path error is close to the
$N^{-1/2}$ fluctuation scale. Each curve averages 24 simulated paths.}
\label{fig:large_population_approximation}
\end{figure}
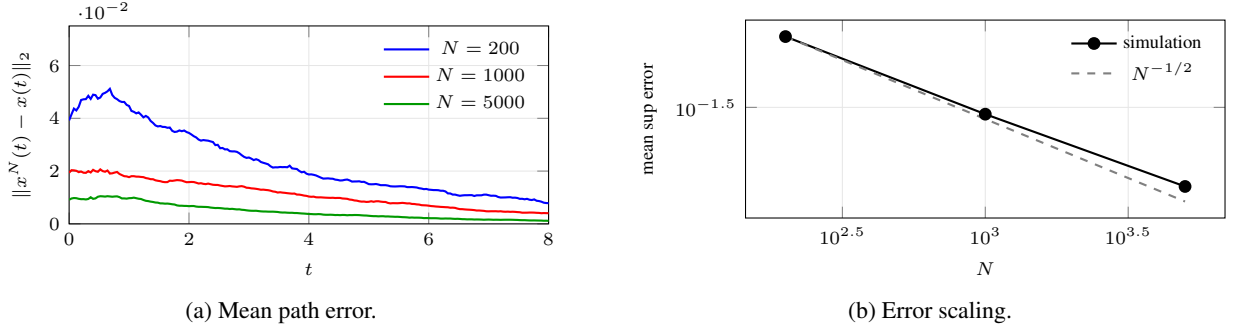

Figure~\ref{fig:large_population_approximation} displays this
large-population approximation. For $N=200$, stochastic switching
creates visible deviations from the ODE path. At $N=1000$ and
$N=5000$, the paths concentrate much more tightly around the
deterministic trajectory. The log-log comparison shows that the
finite-population error decreases at approximately the
$N^{-1/2}$ scale, which is the expected order for aggregate
fluctuations in a large population.

The final experiment introduces a symmetric switching cost $c=0.055$.
The switching intensity becomes
$\lambda_{ij}(x)=\kappa x_j[\rho_j(x)-\rho_i(x)-c]_+$.
The cost does not change the payoff field, but it changes which payoff
differences are actionable. Hence the dynamics need not eliminate all
raw payoff differences; they eliminate only net profitable deviations.

\begin{figure}[t]
\centering
\begin{subfigure}[t]{0.48\textwidth}
\centering
\begin{tikzpicture}
\begin{axis}[
width=\linewidth,
height=4.2cm,
xlabel={$t$},
ylabel={$x_i(t)$},
xmin=0,xmax=18,
ymin=0,ymax=0.65,
tick label style={font=\scriptsize},
label style={font=\scriptsize},
grid=both,
grid style={line width=.1pt,draw=gray!20}
]
\addplot[blue,thick] table[x=t,y=xfree1,col sep=comma] {figures/switching_costs.csv};
\addplot[red,thick] table[x=t,y=xfree2,col sep=comma] {figures/switching_costs.csv};
\addplot[green!60!black,thick] table[x=t,y=xfree3,col sep=comma] {figures/switching_costs.csv};
\addplot[purple,thick] table[x=t,y=xfree4,col sep=comma] {figures/switching_costs.csv};
\addplot[blue,dashed,thick] table[x=t,y=xcost1,col sep=comma] {figures/switching_costs.csv};
\addplot[red,dashed,thick] table[x=t,y=xcost2,col sep=comma] {figures/switching_costs.csv};
\addplot[green!60!black,dashed,thick] table[x=t,y=xcost3,col sep=comma] {figures/switching_costs.csv};
\addplot[purple,dashed,thick] table[x=t,y=xcost4,col sep=comma] {figures/switching_costs.csv};
\node[anchor=north west,font=\scriptsize] at (rel axis cs:0.02,0.98) {solid: no cost};
\node[anchor=north west,font=\scriptsize] at (rel axis cs:0.02,0.88) {dashed: cost};
\end{axis}
\end{tikzpicture}
\caption{Mass trajectories.}
\end{subfigure}
\hfill
\begin{subfigure}[t]{0.48\textwidth}
\centering
\begin{tikzpicture}
\begin{axis}[
width=\linewidth,
height=4.2cm,
xlabel={$t$},
ylabel={payoff gap},
xmin=0,xmax=18,
ymin=0,ymax=0.75,
tick label style={font=\scriptsize},
label style={font=\scriptsize},
legend style={font=\scriptsize,draw=none,fill=none,at={(0.98,0.98)},anchor=north east},
grid=both,
grid style={line width=.1pt,draw=gray!20}
]
\addplot[black,thick] table[x=t,y=range_free,col sep=comma] {figures/switching_costs.csv};
\addlegendentry{range, no cost}
\addplot[red,thick] table[x=t,y=range_cost,col sep=comma] {figures/switching_costs.csv};
\addlegendentry{range, cost}
\addplot[blue,thick] table[x=t,y=net_gain,col sep=comma] {figures/switching_costs.csv};
\addlegendentry{max net gain}
\addplot[gray,dotted,thick] table[x=t,y=cost,col sep=comma] {figures/switching_costs.csv};
\addlegendentry{cost}
\end{axis}
\end{tikzpicture}
\caption{Payoff gaps.}
\end{subfigure}
\caption{Switching costs and constrained stability. Without costs,
payoff differences vanish. With switching cost $c=0.055$, the raw
payoff range remains positive, but the maximum net profitable
deviation converges to zero. The terminal state is stable because
remaining payoff gains are too small to justify switching.}
\label{fig:switching_costs_numerics}
\end{figure}
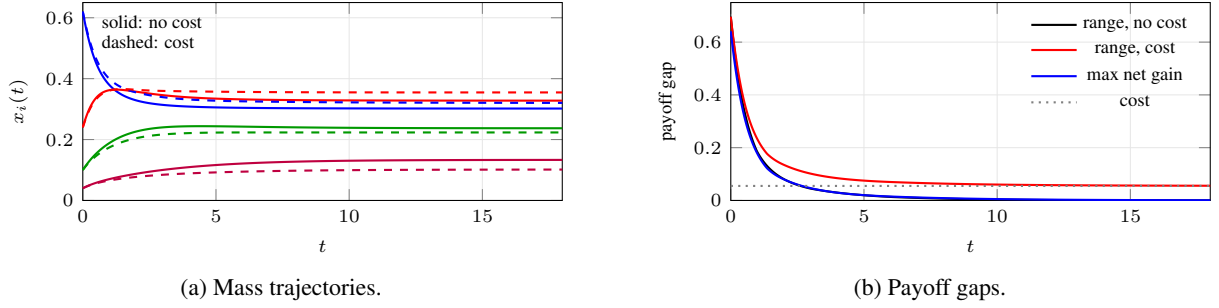

Figure~\ref{fig:switching_costs_numerics} illustrates the constrained
equilibrium logic. In the costless case, payoff densities are equalized.
With switching costs, the dynamics stop earlier: the raw payoff range is
not eliminated, yet the maximum net gain from any admissible switch
approaches zero. Numerically, the constrained steady state therefore
matches the theory in Section~\ref{sec:acceptance_rules}. Frictions can
sustain payoff differences that would be unstable under unrestricted
Wardrop or exit-and-join equilibrium.

\section{Noncooperative Cooperation vs.\ Cooperative Noncooperation}
\label{sec:cooperation_noncooperation}

This work reveals a structural duality between cooperation and
noncooperation. Cooperative structures form the substrate:
players organize into coalitions to accomplish tasks that generate
mutual gains. Coalition formation creates potential surplus and
admits Pareto improvements. In this sense, cooperation defines
the space of feasible collective value.

Yet this cooperative substrate is governed by a noncooperative
layer. Players retain unilateral strategic autonomy: if a more
advantageous coalition becomes available, they may exit and join
another group. Thus cooperation is sustained only insofar as it
remains stable against individual deviation. Noncooperation defines
the stability constraints imposed on cooperative arrangements.

The duality arises because cooperation defines potential surplus,
whereas noncooperation defines stability under deviation. An outcome may
be efficient but unstable; conversely, a stable equilibrium may be
inefficient. The tension between efficiency and stability is the core of
the cooperation--noncooperation duality.

There is therefore an inherent interplay between the two layers.
Cooperation induces noncooperation: precisely because agents seek
better cooperative outcomes, they must retain the option to deviate
unilaterally. At the same time, noncooperative mobility can improve
cooperation: unilateral deviations reconfigure coalition structures,
sometimes increasing overall efficiency or enabling superior
collective performance. Improvements may be local, but they can also
propagate globally through structural reorganization.

The relationship is thus dualistic rather than contradictory.
Cooperation and noncooperation are intertwined mechanisms of the
same system. One operates at the structural level (coalition value
creation), the other at the strategic level (individual incentive
compatibility). Their interaction admits a natural dual perspective:
the cooperative problem specifies feasible collective structures,
while the noncooperative problem governs admissible deviations
within that feasible set.

Historically, this tension has appeared in debates over collective
action and governance. Hardin's tragedy of the commons and Olson's
logic of collective action show how individually rational behavior can
undermine shared resources and public goods
\cite{Hardin1968,Olson1965}. A pure cooperation view, by contrast,
assumes that common interest suffices to sustain collective action.

Ostrom's common-pool resource framework asked a question closely
aligned with the present framework: how can individually rational
agents sustain cooperation without centralized enforcement
\cite{Ostrom1990}? Her empirical studies demonstrated that cooperation
and noncooperation coexist within structured rule systems. Agents
remain strategic and capable of deviation, yet institutions reshape the
incentive landscape so that cooperative behavior becomes
self-enforcing. In our terminology, institutions restrict the
admissible deviation set and alter the effective payoff functional.
They mediate the cooperation--noncooperation duality by aligning
collective efficiency with individual stability.

Our framework differs in emphasis but parallels her insight.
Rather than assuming externally imposed rules, we show that
decentralized exit--and--join dynamics can endogenously reconcile
cooperative value creation with noncooperative stability.
Cooperation need not be imposed from above; it can emerge from below
through structured strategic mobility.

Classical coordination games such as the Battle of the Sexes
illustrate this duality. The objective is cooperative alignment,
yet equilibrium selection proceeds through noncooperative reasoning.
As emphasized by Adam Smith, decentralized self-interest can,
under appropriate structure, generate cooperative order
\cite{Smith1776}.

Cooperation defines what is desirable.
Noncooperation determines what is sustainable.
The equilibrium of a multi-agent system lies at their intersection.

\section{Exit-and-Join and Evolution}
\label{sec:evolutionary_interpretation}

This section interprets the preceding model as an evolutionary system on
coalition structures. The interpretation does not introduce a new
equilibrium concept. Rather, it shows that the same payoff density that
allocates cooperative value also acts as a selection index for coalition
growth, in the sense of evolutionary population dynamics
\cite{MaynardSmithPrice1973,MaynardSmith1982,TaylorJonker1978,
Weibull1995,HofbauerSigmund1998,vincent2005evolutionary,
sandholm2010population,Nowak2006,zhuTembineBasar2011evolutionaryMAC,
hayelZhu2015evolutionaryPoisson,liuZhaoZhu2021herd}.

\subsection{Marginal Contribution Densities and Fitness}

In the continuum cooperative model, a coalition is represented by a
restricted measure $\mu_S$, and its value is generated by a functional
$V$ on measures. The local object that determines payoff is the
functional derivative
\[
\frac{\delta V}{\delta \mu}(\mu_S)(i),
\]
which gives the marginal contribution density of an infinitesimal mass
at player location $i$. The Aumann--Shapley value averages this marginal
contribution along a participation path
\cite{aumann1975values,aumann2015values}:
\[
\phi_i^{AS}(v)
=
\int_0^1
\frac{\delta V}{\delta \mu}(\mu_{S_\lambda})(i)
\,d\lambda .
\]
Thus value allocation is determined by accumulated marginal productivity
along a coalition-growth path.

In evolutionary population models, the analogous local object is the
$G$-function: the per-capita growth rate of a rare type inserted into a
resident population \cite{vincent2000evolution,vincent2005evolutionary}.
Written in measure terms, both objects are
directional marginal values. The cooperative model uses the marginal
value to allocate surplus; the evolutionary model uses it to determine
growth. The correspondence is direct: coalition measures play the role
of population distributions, the cooperative value functional plays the
role of a fitness landscape, and the marginal contribution density plays
the role of a $G$-function. This explains why the payoff density
$\rho_i(x)$ in the exit-and-join model can be interpreted as a coalition
fitness index.

\subsection{Selection Dynamics}

Exit-and-join dynamics turn payoff comparisons into population movement.
If agents leave lower-payoff coalitions for higher-payoff coalitions,
then coalition mass is reallocated in the direction of higher marginal
productivity. Under the pairwise payoff-difference rule
\[
\lambda_{ij}(x)=\kappa x_j[\rho_j(x)-\rho_i(x)]_+,
\]
the mean-field equation reduces to the replicator form familiar from
evolutionary game theory
\cite{TaylorJonker1978,Weibull1995,HofbauerSigmund1998,
sandholm2010population,vincent2005evolutionary},
\[
\dot x_i=\kappa x_i\bigl(\rho_i(x)-\bar\rho(x)\bigr),
\qquad
\bar\rho(x)=\sum_j x_j\rho_j(x).
\]
Coalitions with payoff above the population average expand, while those
below the average contract.

The replicator equation is only one specification. Other
incentive-compatible switching rules preserve the same principle: mass
moves toward higher payoff, with speed determined by the rate rule
\cite{Weibull1995,HofbauerSigmund1998,sandholm2010population}.

\subsection{Micro Incentives and Macro Evolution}

The evolutionary behavior is generated by agents, not by coalitions.
Coalitions do not choose to grow or shrink. Agents compare their current
payoff with feasible alternatives and move when a profitable destination
is admissible. Aggregating these decentralized choices yields the
deterministic mass dynamics derived in
Section~\ref{sec:continuum_exit_join}, paralleling classical
law-of-large-numbers limits for Markov population processes
\cite{ethier2009markov}.

This produces a two-level structure. Within each coalition, agents
cooperate to generate transferable surplus. Across coalitions, the
resulting payoff densities determine which coalitions attract mass.
Cooperation therefore supplies the value functional, while mobility
creates selection among cooperative arrangements, connecting coalition
formation dynamics with population adjustment dynamics
\cite{apt2009generic,konishi2003coalition,MaynardSmith1982,
Weibull1995,sandholm2010population}.

In this sense, Darwinian selection is an emergent property of
exit-and-join incentives. It does not require a planner or an
organization-level optimization rule. It requires only that agents can
move in response to relative payoff differences.

\subsection{Equilibrium and Constrained Stability}

The evolutionary interpretation also clarifies stability. In the
unconstrained model, an exit-and-join equilibrium is a state in which no
positive-mass group of agents can improve by moving to another coalition.
Under incentive-compatible and strictly payoff-responsive switching
rates, this condition coincides with stationarity of the mean-field
dynamics. Equivalently, no active transition flux remains in a
payoff-improving direction.

With switching costs or acceptance constraints, stability becomes
conditional on the feasible deviation set. Payoff differences may persist
because a profitable move is too costly or because the destination
coalition rejects marginal entry. The relevant notion is therefore
constrained evolutionary stability: no admissible positive-mass deviation
can invade the current coalition structure. This interpretation is close
to the invasion-resistance logic of evolutionary stability and to
viability and mutational viewpoints in constrained dynamical systems
\cite{MaynardSmithPrice1973,MaynardSmith1982,TaylorJonker1978,
Aubin1991Viability,Aubin1998Mutational}.

\section{Conclusion}
\label{sec:conclusion}

This paper developed a continuum model of coalition formation in
nonatomic cooperative games. It extended the Aumann--Shapley and
Aumann--Dr\`eze values to finite coalition structures by treating each
coalition as a restricted nonatomic game and assigning payoff densities
through infinitesimal marginal contributions
\cite{aumann1975values,aumann2015values,aumann1974cooperative,
AumannDreze1974}.

The paper then derived exit-and-join dynamics from decentralized
switching rules. A finite-population approximation yields a deterministic
mean-field ODE for coalition masses, and payoff-difference switching
recovers replicator dynamics as a special case. The associated
exit-and-join equilibrium rules out profitable positive-mass deviations
and, under incentive-compatible and strictly payoff-responsive switching
rates, coincides with stationarity of the mass dynamics
\cite{ethier2009markov,sandholm2010population,vincent2005evolutionary,
zhuTembineBasar2011evolutionaryMAC,hayelZhu2015evolutionaryPoisson,
liuZhaoZhu2021herd}.

For mass-based cooperative games, the dynamics admit a Lyapunov function
derived from aggregate cooperative surplus. Under strict concavity and
the stated regularity and responsiveness assumptions, trajectories
converge globally to the unique equilibrium. The same equilibrium
condition is equivalent to Wardrop equilibrium in the induced nonatomic
population game and admits a variational inequality formulation
\cite{Wardrop1952,FacchineiPang2003,panZhu2022poisonedWardrop,
panLiZhu2022resilienceTraffic}.

Switching costs and acceptance rules restrict feasible deviations and
lead to constrained equilibria. In that setting, payoff differences can
persist because mobility is no longer unrestricted. The resulting
quasi-variational inequality formulation captures how institutional
frictions, congestion, and endogenous entry rules alter coalition
stability \cite{FacchineiPang2003,Aubin1991Viability}.

Overall, the framework links cooperative value creation, noncooperative
mobility, and evolutionary selection. It provides a common language for
large-population coalition systems in which coalitions produce surplus
cooperatively while agents reallocate through individual incentives.
Future work may study stochastic stability, finite-population error
bounds, and applications to labor markets, platform migration, distributed
task allocation, and reconfigurable multi-agent systems
\cite{ethier2009markov,smyrnakis2019game,hamed2023distributed}.

\bibliographystyle{unsrt}
\bibliography{ref}

@book{ethier2009markov,
  title={Markov processes: characterization and convergence},
  author={Ethier, Stewart N and Kurtz, Thomas G},
  year={2009},
  publisher={John Wiley \& Sons}
}

@article{zhuTembineBasar2011evolutionaryMAC,
  title={Evolutionary Games for Multiple Access Control},
  author={Zhu, Quanyan and Tembine, Hamidou and Ba{\c{s}}ar, Tamer},
  journal={arXiv preprint arXiv:1103.2496},
  year={2011},
  eprint={1103.2496},
  archivePrefix={arXiv},
  primaryClass={cs.GT},
  doi={10.48550/arXiv.1103.2496},
  url={https://arxiv.org/abs/1103.2496}
}

@article{hayelZhu2015evolutionaryPoisson,
  title={Evolutionary Poisson Games for Controlling Large Population Behaviors},
  author={Hayel, Yezekael and Zhu, Quanyan},
  journal={arXiv preprint arXiv:1503.08085},
  year={2015},
  eprint={1503.08085},
  archivePrefix={arXiv},
  primaryClass={cs.GT},
  doi={10.48550/arXiv.1503.08085},
  url={https://arxiv.org/abs/1503.08085}
}

@article{liuZhaoZhu2021herd,
  title={Herd Behaviors in Epidemics: A Dynamics-Coupled Evolutionary Games Approach},
  author={Liu, Shutian and Zhao, Yuhan and Zhu, Quanyan},
  journal={arXiv preprint arXiv:2106.08998},
  year={2021},
  eprint={2106.08998},
  archivePrefix={arXiv},
  primaryClass={cs.GT},
  doi={10.48550/arXiv.2106.08998},
  url={https://arxiv.org/abs/2106.08998}
}

@article{panZhu2022poisonedWardrop,
  title={On Poisoned Wardrop Equilibrium in Congestion Games},
  author={Pan, Yunian and Zhu, Quanyan},
  journal={arXiv preprint arXiv:2209.00094},
  year={2022},
  eprint={2209.00094},
  archivePrefix={arXiv},
  primaryClass={cs.GT},
  doi={10.48550/arXiv.2209.00094},
  url={https://arxiv.org/abs/2209.00094}
}

@article{panLiZhu2022resilienceTraffic,
  title={On the Resilience of Traffic Networks under Non-Equilibrium Learning},
  author={Pan, Yunian and Li, Tao and Zhu, Quanyan},
  journal={arXiv preprint arXiv:2210.03214},
  year={2022},
  eprint={2210.03214},
  archivePrefix={arXiv},
  primaryClass={eess.SY},
  doi={10.48550/arXiv.2210.03214},
  url={https://arxiv.org/abs/2210.03214}
}

@article{apt2009generic,
  title={A generic approach to coalition formation},
  author={Apt, Krzysztof R and Witzel, Andreas},
  journal={International game theory review},
  volume={11},
  number={03},
  pages={347--367},
  year={2009},
  publisher={World Scientific}
}

@article{Shapley1953,
  author  = {Shapley, Lloyd S.},
  title   = {A Value for n-Person Games},
  journal = {Contributions to the Theory of Games},
  volume  = {2},
  pages   = {307--317},
  year    = {1953},
  publisher = {Princeton University Press}
}

@inproceedings{hamed2023distributed,
  title={Distributed learning dynamics for coalitional games},
  author={Hamed, Aya and Shamma, Jeff S},
  booktitle={2023 62nd IEEE Conference on Decision and Control (CDC)},
  pages={5020--5025},
  year={2023},
  organization={IEEE}
}

@article{aumann1974cooperative,
  title={Cooperative games with coalition structures},
  author={Aumann, Robert J and Dreze, Jacques H},
  journal={International Journal of game theory},
  volume={3},
  number={4},
  pages={217--237},
  year={1974},
  publisher={Springer}
}

@article{hart1996bargaining,
  title={Bargaining and value},
  author={Hart, Sergiu and Mas-Colell, Andreu},
  journal={Econometrica: Journal of the Econometric Society},
  pages={357--380},
  year={1996},
  publisher={JSTOR}
}

@article{faigle2001computation,
  title={On the computation of the nucleolus of a cooperative game},
  author={Faigle, Ulrich and Kern, Walter and Kuipers, Jeroen},
  journal={International Journal of Game Theory},
  volume={30},
  number={1},
  pages={79--98},
  year={2001},
  publisher={Springer}
}

@article{AumannDreze1974,
  author  = {Aumann, Robert J. and Dr{\`e}ze, Jacques H.},
  title   = {Cooperative Games with Coalition Structures},
  journal = {International Journal of Game Theory},
  volume  = {3},
  number  = {4},
  pages   = {217--237},
  year    = {1974},
  doi     = {10.1007/BF01766876}
}

@article{touati2021bayesian,
  title={A Bayesian Monte Carlo method for computing the Shapley value: Application to weighted voting and bin packing games},
  author={Touati, Sofiane and Radjef, Mohammed Said and others},
  journal={Computers \& Operations Research},
  volume={125},
  pages={105094},
  year={2021},
  publisher={Elsevier}
}

@book{Aubin1991Viability,
  author    = {Aubin, Jean-Pierre},
  title     = {Viability Theory},
  publisher = {Birkh{\"a}user},
  year      = {1991}
}

@book{Aubin1998Mutational,
  author    = {Aubin, Jean-Pierre},
  title     = {Mutational and Morphological Analysis: Tools for Shape Evolution and Morphogenesis},
  publisher = {Birkh{\"a}user},
  year      = {2001},
  doi       = {10.1007/978-1-4612-1576-9}
}

@article{smyrnakis2019game,
  title={Game-theoretic learning and allocations in robust dynamic coalitional games},
  author={Smyrnakis, Michalis and Bauso, Dario and Tembine, Hamidou},
  journal={SIAM Journal on Control and Optimization},
  volume={57},
  number={4},
  pages={2902--2923},
  year={2019},
  publisher={SIAM}
}

@book{sandholm2010population,
  title={Population games and evolutionary dynamics},
  author={Sandholm, William H},
  year={2010},
  publisher={MIT press}
}

@article{konishi2003coalition,
  title={Coalition formation as a dynamic process},
  author={Konishi, Hideo and Ray, Debraj},
  journal={Journal of Economic theory},
  volume={110},
  number={1},
  pages={1--41},
  year={2003},
  publisher={Elsevier}
}

@article{filar2000dynamic,
  title={Dynamic cooperative games},
  author={Filar, Jerzy A and Petrosjan, Leon A},
  journal={International Game Theory Review},
  volume={2},
  number={01},
  pages={47--65},
  year={2000},
  publisher={World Scientific}
}

@article{bauso2009robust,
  title={Robust dynamic cooperative games},
  author={Bauso, Dario and Timmer, Judith},
  journal={International Journal of Game Theory},
  volume={38},
  number={1},
  pages={23--36},
  year={2009},
  publisher={Springer}
}

@article{aumann1975values,
  title={Values of markets with a continuum of traders},
  author={Aumann, Robert J},
  journal={Econometrica: Journal of the Econometric Society},
  pages={611--646},
  year={1975},
  publisher={JSTOR}
}

@book{vincent2005evolutionary,
  title={Evolutionary game theory, natural selection, and Darwinian dynamics},
  author={Vincent, Thomas L and Brown, Joel S},
  year={2005},
  publisher={Cambridge University Press}
}

@article{vincent2000evolution,
  title={Evolution and control system design. The evolutionary game},
  author={Vincent, Thomas L and Vincent, Tania LS},
  journal={IEEE Control Systems Magazine},
  volume={20},
  number={5},
  pages={20--35},
  year={2000},
  publisher={IEEE}
}

@book{aumann2015values,
  title={Values of non-atomic games},
  author={Aumann, Robert J and Shapley, Lloyd S},
  year={2015},
  publisher={Princeton University Press}
}

@article{winter2002shapley,
  title={The shapley value},
  author={Winter, Eyal},
  journal={Handbook of game theory with economic applications},
  volume={3},
  pages={2025--2054},
  year={2002},
  publisher={Elsevier}
}

@article{Wardrop1952,
  author  = {Wardrop, John Glen},
  title   = {Some Theoretical Aspects of Road Traffic Research},
  journal = {Proceedings of the Institution of Civil Engineers},
  volume  = {1},
  number  = {3},
  pages   = {325--362},
  year    = {1952}
}

@book{FacchineiPang2003,
  author    = {Facchinei, Francisco and Pang, Jong-Shi},
  title     = {Finite-Dimensional Variational Inequalities and Complementarity Problems},
  publisher = {Springer},
  address   = {New York},
  year      = {2003}
}

@article{MaynardSmithPrice1973,
  author  = {Maynard Smith, John and Price, George R.},
  title   = {The Logic of Animal Conflict},
  journal = {Nature},
  volume  = {246},
  pages   = {15--18},
  year    = {1973},
  doi     = {10.1038/246015a0}
}

@book{MaynardSmith1982,
  author    = {Maynard Smith, John},
  title     = {Evolution and the Theory of Games},
  publisher = {Cambridge University Press},
  address   = {Cambridge},
  year      = {1982}
}

@article{TaylorJonker1978,
  author  = {Taylor, Peter D. and Jonker, Leo B.},
  title   = {Evolutionarily Stable Strategies and Game Dynamics},
  journal = {Mathematical Biosciences},
  volume  = {40},
  number  = {1--2},
  pages   = {145--156},
  year    = {1978},
  doi     = {10.1016/0025-5564(78)90077-9}
}

@book{Weibull1995,
  author    = {Weibull, J{\"o}rgen W.},
  title     = {Evolutionary Game Theory},
  publisher = {MIT Press},
  address   = {Cambridge, MA},
  year      = {1995}
}

@book{HofbauerSigmund1998,
  author    = {Hofbauer, Josef and Sigmund, Karl},
  title     = {Evolutionary Games and Population Dynamics},
  publisher = {Cambridge University Press},
  address   = {Cambridge},
  year      = {1998}
}

@book{Nowak2006,
  author    = {Nowak, Martin A.},
  title     = {Evolutionary Dynamics: Exploring the Equations of Life},
  publisher = {Harvard University Press},
  address   = {Cambridge, MA},
  year      = {2006}
}

@article{Hardin1968,
  author  = {Hardin, Garrett},
  title   = {The Tragedy of the Commons},
  journal = {Science},
  volume  = {162},
  number  = {3859},
  pages   = {1243--1248},
  year    = {1968},
  doi     = {10.1126/science.162.3859.1243}
}

@book{Olson1965,
  author    = {Olson, Mancur},
  title     = {The Logic of Collective Action: Public Goods and the Theory of Groups},
  publisher = {Harvard University Press},
  address   = {Cambridge, MA},
  year      = {1965}
}

@book{Ostrom1990,
  author    = {Ostrom, Elinor},
  title     = {Governing the Commons: The Evolution of Institutions for Collective Action},
  publisher = {Cambridge University Press},
  address   = {Cambridge},
  year      = {1990}
}

@book{Smith1776,
  author    = {Smith, Adam},
  title     = {An Inquiry into the Nature and Causes of the Wealth of Nations},
  publisher = {W. Strahan and T. Cadell},
  address   = {London},
  year      = {1776}
}

\end{document}